\documentclass[draftcls,onecolumn,twoside,american]{IEEEtran}
\usepackage{amsmath,epsfig}
\usepackage{inputenc}
\usepackage[T1]{fontenc}
\usepackage[french,english]{babel}
\usepackage{amssymb,amsthm,color,bbm}
\usepackage{dsfont}
\usepackage{graphicx}
\usepackage{subfig}
\usepackage{version}
\usepackage{ifthen}
\usepackage{url}
\usepackage{cite}
\usepackage{algorithm,algorithmic}

\usepackage{amsbsy}

\definecolor{grey}{rgb}{0.75,0.75,0.75}
\definecolor{orange}{rgb}{1.0,0.5,0.5}
\definecolor{brown}{rgb}{0.5,0.25,0.0}
\definecolor{pink}{rgb}{1.0,0.5,0.5}

\makeatother

\newcommand{\bigN}{\mathbb{N}}
\newcommand{\bigX}{\mathbb{X}}
\newcommand{\bigY}{\mathbb{Y}}
\newcommand{\bigR}{\mathbb{R}}
\newcommand{\bigP}{\mathbb{P}}
\newcommand{\bigZ}{\mathbb{Z}}
\newcommand{\bigKT}{\mathbb{KT}}
\newcommand{\argmax}{\operatornamewithlimits{argmax}} 
\newcommand{\argmin}{\operatornamewithlimits{argmin}} 

\newcommand{\one}{\mathbf{1}}

\newtheorem{theo}{Theorem} 
 
\newtheorem{lem}{Lemma} 
\newtheorem{pr}{Proposition}

\newtheorem{defi}{Definition}

\usepackage{verbatim}
\usepackage{setspace}

\newcommand{\CE}[4][]
{
\ifthenelse{\equal{#1}{}}{\mathbb{E}^{#2}_{#3}\left[#4\right]}{\mathbb{E}^{#2}_{#3}\left[#4\middle | #1\right]}
}

\def\Xset{\mathbb{X}}

\newcommand{\XinitIS}[2][]{\ifthenelse{\equal{#1}{}}{\ensuremath{\rho_{#2}}}{\ensuremath{\check{\rho}_{#2}}}}

\newcommand{\filt}[2][]%
{
\ifthenelse{\equal{#1}{}}{\ensuremath{\phi_{#2}}}%
{\ifthenelse{\equal{#1}{hat}}{\ensuremath{\phi^{N}_{#2}}}
{\ifthenelse{\equal{#1}{tilde}}{\ensuremath{\tilde{\phi}^{N}_{#2}}}
{\ifthenelse{\equal{#1}{tar}}{\ensuremath{\phi^{N,\mathrm{t}}_{#2}}}
{\ifthenelse{\equal{#1}{aux}}{\ensuremath{\phi^{N,\mathrm{a}}_{#2}}}
}
}
}
}
}
\newcommand{\mcbf}[2][]{%
\ifthenelse{\equal{#1}{}}{\overline{\mathcal{F}}_{#2}}{\overline{\mathcal{F}}_{#2}^{(#1)}}%
}

\newcommand{\unfilt}[2][]%
{
\ifthenelse{\equal{#1}{}}{\ensuremath{\gamma_{#2}}}%
{\ifthenelse{\equal{#1}{hat}}{\ensuremath{\gamma^{N}_{#2}}}
{\ifthenelse{\equal{#1}{tilde}}{\ensuremath{\tilde{\gamma}^{N}_{#2}}}
{\ifthenelse{\equal{#1}{tar}}{\ensuremath{\gamma^{N,\mathrm{t}}_{#2}}}
{\ifthenelse{\equal{#1}{aux}}{\ensuremath{\gamma^{N,\mathrm{a}}_{#2}}}
}
}
}
}
}
\newcommand{\CPE}[3][]
{\ifthenelse{\equal{#1}{}}{\mathbb{E}\left[\left. #2 \, \right| #3 \right]}{\mathbb{E}_{#1}\left[\left. #2 \, \right| #3 \right]}}

\newcommand{\Xsigma}[1][]%
{%
\ifthenelse{\equal{#1}{}}{\ensuremath{\mathcal{B}(\Xset)}}{\ensuremath{\mathcal{B}(\Xset^{#1})}}
}
\newcommand{\sumwght}[2][]{%
\ifthenelse{\equal{#1}{}}{\ensuremath{\Omega_{#2}}}{\ensuremath{\Omega_{#2}^{(#1)}}}}

\newcommand{\adjfunc}[4][]
{\ifthenelse{\equal{#1}{}}{\ifthenelse{\equal{#4}{}}{\vartheta_{#2}}{\vartheta_{#2}(#4)}}
{\ifthenelse{\equal{#1}{smooth}}{\ifthenelse{\equal{#4}{}}{\tilde{\vartheta}_{#2}}{\tilde{\vartheta}_{#2}(#4)}}
{\ifthenelse{\equal{#1}{fully}}{\ifthenelse{\equal{#4}{}}{\vartheta^\star_{#2}}{\vartheta^\star_{#2}(#4)}}{\mathrm{erreur}}}}}
\newcommand{\chunk}[4][]%
{\ifthenelse{\equal{#1}{}}{\ensuremath{{#2}_{#3:#4}}}{\ensuremath{#2^#1}_{#3:#4}}
}
\newcommand{\kiss}[3][]
{\ifthenelse{\equal{#1}{}}{p_{#2}}
{\ifthenelse{\equal{#1}{fully}}{p^{\star}_{#2}}
{\ifthenelse{\equal{#1}{smooth}}{\tilde{r}_{#2}}{\mathrm{erreur}}}}}

\newcommand{\Kiss}[3][]
{\ifthenelse{\equal{#1}{}}{P_{#2}}
{\ifthenelse{\equal{#1}{fully}}{P^{\star}_{#2}}
{\ifthenelse{\equal{#1}{smooth}}{\tilde{R}_{#2}}{\mathrm{erreur}}}}}

\newcommand{\post}[3][]%
{
\ifthenelse{\equal{#1}{}}{\ensuremath{\phi_{#2|#3}}}%
{\ifthenelse{\equal{#1}{hat}}{\ensuremath{\phi^{\mathrm{FFBS},N}_{#2|#3}}}
{\ifthenelse{\equal{#1}{tilde}}{\ensuremath{\phi^{\mathrm{FFBSi},N}_{#2|#3}}}
{\ifthenelse{\equal{#1}{tar}}{\ensuremath{\phi^{N,\mathrm{t}}_{#2|#3}}}}
}
}
}

\newcounter{hypA}

\begin{document}

\title{Context Tree Estimation in Variable Length Hidden Markov Models}

\author{Thierry~Dumont\\ Universit\'e de Paris sud XI\\E-mail: thierry.dumont@math.u-psud.fr
\thanks{This work was supported by ID Services, 22/24 rue Jean Rostand 91400 Orsay, FRANCE}}

\maketitle

\begin{abstract}
We address the issue of context tree estimation in variable length hidden Markov models. We propose an estimator of the context tree of the hidden Markov process which needs no prior upper bound on the depth of the context tree. We prove that the estimator is strongly consistent. This uses information-theoretic mixture inequalities in the spirit of \cite{Fin90,GB03}. We propose an algorithm to efficiently compute the estimator and provide simulation studies to support our result.
\end{abstract}


\IEEEpeerreviewmaketitle

\section{Introduction}

A variable length hidden Markov model (VLHMM) is a bivariate stochastic process $(X_{n},Y_{n})_{n \geq 0}$ where $(X_{n})_{n \geq 0}$ (the state sequence) is a variable length Markov chain (VLMC) in a state space $\bigX$ and, conditionally on $(X_{n})_{n \geq 0}$, $(Y_{n})_{n \geq 0}$ is a sequence of independent variables in a state space $\bigY$ such that the conditional distribution of $Y_{n}$ given the state sequence (called the emission distribution) depends on $X_{n}$ only. Such processes fall into the general framework of latent variable processes, and reduce to hidden Markov models (HMM) in case the state sequence is a Markov chain. Latent variable processes are used as a flexible tool to model dependent non-Markovian time series, and the statistical problem is to estimate the  parameters of the distribution when only $(Y_{n})_{n \geq 0}$ is observed. We will consider in this paper the case where the hidden process may take only a fixed and known number of values, that is the case where the state space $\bigX$ is finite with known cardinality $k$. 

The dependence structure of a latent variable process  is driven by that of the hidden process $(X_{n})_{n \geq 0}$, which is assumed here to be a variable length Markov chain (VLMC). Such processes were first introduced by Rissanen in \cite{Ris83} as a flexible and parsimonious modelization tool for data compression, approximating  Markov chains of finite orders. Recall that a Markov process of order $d$  is  such that the conditional distribution of $X_{n}$ given all past values depends only on the $d$ previous ones $X_{n-1},\ldots,X_{n-d}$. 
But different past values may lead to identical conditional distributions, so that all $k^{d}$ possible past values are not needed to describe the distribution of the process. A VLMC is such that the probability of the present state depends only on a finite part of the past, and the length of this relevant portion, called context, is a function of the past itself.
No context may be a proper postfix of any other context, so that the set of all contexts may be represented as a rooted labelled tree. This set is called the context tree of the VLMC. 

Variable length hidden Markov models appear for the first time, to our knowledge, in movement analysis \cite{WangZhou}, \cite{Wang}. Human movement analysis is the interpretation of movements as sequences of poses. \cite{Wang} analyses the movement through 3D rotations of 19 major joints of human body. Wang and al. then use a VLHMM representation where $X_{n}$ is the pose at time $n$ and $Y_{n}$ is the body position given by the 3D rotations of the 19 major points. They argue that "VLHMM is superior in its efficiency and accuracy of modeling multivariate time-series data with highly-varied dynamics". 

VLHMM could also be used in WIFI based indoor positioning systems (see \cite{evennou}). Here $X_{n}$ is a mobile device position at time $n$ and $Y_{n}$ is the received signal strength (RSS) vector at time $n$. Each component of the RSS vector represents the strength of a signal sent by a WIFI access point. In practice, the aim is to estimate the positions of the device $(X_{n})_{n\ge 0}$ on the basis of the observations $(Y_{n})_{n\ge 0}$. The distribution of $Y_{n}$ given $X_{n} = x$ for any location $x$ is beforehand calibrated for a finite number of locations $(L_{1},...,L_{k})$. A Markov chain on the finite set $(L_{1},...,L_{k})$ is then used to model the sequence of positions $(X_{n})_{n \ge 0}$. Again VLHMM model would lead to efficient and accurate estimation of the device position.

The aim of this paper is to provide a statistical analysis of variable length hidden Markov models and, in particular, to propose a consistent estimator of the context tree of the hidden VLMC
on the basis of the observations $(Y_{n})_{n \geq 0}$ only. 
We consider a parametrized family of VLHMM, and we use a penalized likelihood method
to estimate the context tree of the hidden VLMC. To each possible context tree $\tau$, if $\Theta_{\tau}$ is the set of possible parameters, we define
$$\hat{\tau}_{n} = \argmin\limits_{\tau} \left\{ -\sup\limits_{\theta \in \Theta_{\tau}} \log g_{\theta}(Y_{1:n}) + pen(n,\tau) \right\}$$
where $g_{\theta}(y_{1:n})$ is the density of the distribution of the observation $Y_{1:n}=(Y_{1},\ldots,Y_{n})$ under the parameter $\theta$ with respect to some dominating positive measure, and $pen(n,\tau)$ is a penalty that depends on the number $n$ of observations and the context tree $\tau$. 
Our aim is to find penalties for which the estimator is strongly consistent without any prior  upper bound on the depth of the context tree, and to provide a practical algorithm to compute the estimator.

Context tree estimation for a VLHMM is similar to  order estimation for a  HMM in which the order is defined as the unknown cardinality of the state space $\bigX$. The main difficulty lies in 
the calibration of the penalty, which requires some understanding of the growth of the likelihood ratios (with respect to orders and to the number of observations). In particular cases, the fluctuations of the likelihood ratios may be understood via empirical process theory, see the recent works \cite{Han09} for finite state Markov chains and \cite{HanGas} for independent identically distributed observations. Latent variable models are much more complicated, see for instance \cite{EGCK00} where it is proved in the HMM situation that the likelihood ratio statistics converges to infinity for overestimated order. We thus use an approach based on information theory tools to understand the behavior of likelihood ratios. Such tools have been successfull for HMM order estimation problems and were used in \cite{GB03}, \cite{Fin90} for discrete observations and in \cite{CGG09}  for Poisson emission distributions or Gaussian emission distributions with known variance. Our main result shows that for a penalty of form $C(\tau)\log n$, $\hat{\tau}_{n} $ is strongly consistent, that is converges almost surely to the true unknown context tree. Here, $C(\tau)$ has an explicit formulation but is slightly bigger than $(k-1)|\tau|/2$ which gives the popular BIC penalty. 
We study the important situation of Gaussian emissions with unknown variance, and prove that our consistency theorem holds in this case. 

Computation of the estimator requires computation of the maximum likelihood for all possible context trees. As usual, the EM algorithm may be used to compute the maximum likelihood estimator for the parameters when the context tree is fixed. We then propose an algorithm to compute the estimator, which prevents the exploration of a too large number of context trees. In general the EM algorithm needs to be run several times with different initial values to avoid local extrema traps. In the important situation of Gaussian emissions, we propose a way to choose the initial parameters so that only one run of the EM algorithm is needed. Simulations compare penalized maximum likelihood estimators of the context tree $\tau$ of the hidden VLMC using our penalty and using BIC penalty.

The structure of this paper is the following. Section \ref{sec:not} describes the model and gives the notations.
Section \ref{sec:theo} presents the information theory tools we use, states the main consistency result and applies it to Poisson emission distributions and Gaussian emission distributions with known variance. Section \ref{sec:gauss} proves the result for Gaussian emission distributions with unknown variance. 
In section \ref{sec:algo}, we describe the algorithm to compute the estimator and we give the simulation results. The proofs that are not essential at first reading are detailed in the Appendix.

\section{Basic setting and notation}
\label{sec:not}
Let $\bigX$ be a finite set whose cardinality is denoted by $\left| \bigX \right|=k$, that we identify with $\{1,\ldots,k\}$. Let ${\cal{F}}_{\bigX}$ be the finite collection of subsets of $\bigX$. Let  $\bigY$ be a Polish space endowed with its Borel sigma-field ${\cal{F}}_{\bigY}$. We will work on the measurable space $(\Omega,\cal{F})$ with $\Omega=(\bigX \times \bigY)^{\bigN}$ and $\cal{F}=({\cal{F}}_{\bigX}\otimes {\cal{F}}_{\bigY})^{\otimes \bigN}$
. 
\subsection{Context trees and variable length Markov chains}
A string $s=x_{k} x_{k+1} ...x_{l} \ \in \bigX^{l-k+1} $ is denoted by $x_{k:l}$ and its length is then $l(s)=l-k+1$. We call letters of $s$ its components $x_{i}, \ i=k,\ldots,l$. The concatenation of the strings $u$ and $v$ is denoted by $uv$. 
A string $v$ is a \textit{postfix} of a string $s$ if there exists a string $u$ such that $s=uv$.\\
A set $\tau$ of strings and possibly semi-infinite sequences is called a \textit{tree} if the following \textit{tree property} holds : no $s\in \tau$ is postfix of any other $s' \in \tau$.
A tree $\tau$ is \textit{irreducible} if no element $s \in \tau$ can be replaced by a postfix without violating the tree property. It is \textit{complete} if each node except the leaves has $\left|\bigX\right|$ children exactly. We denote by $d(\tau)$ the depth of $\tau$ : $d(\tau) = \max \left\{ l(s) \ \big| \ s\in\tau \right\}$.

Let now $Q$ be the distribution of an ergodic stationary process $(X_{n})_{n\in\bigZ}$ on $(\bigX^{\bigZ},{\cal{F}}_{\bigX}^{\otimes \bigZ})$, and for any $m\leq n$ and any $x_{m:n}$ in $\bigX^{n-m+1}$, write $Q(x_{m:n})$ for $Q(X_{0:n-m}=x_{m:n})$.

\begin{defi}\label{defi VLMC}  Let $\tau$ be a tree. $\tau$ is called a $Q$-adapted context tree if for any string $s$ in $\tau$ such that $Q(s)>0$:
\begin{equation}
\label{eqQadapt}
 \forall x_{0}\in \bigX ,  Q(X_{0} = x_{0}  \big|   X_{-\infty:-1} = x_{-\infty:-1})=  Q(X_{0} = x_{0}  \big| X_{-l(s):-1} = s)
\end{equation} whenever $s$ is postfix of the semi infinite sequence  $x_{-\infty:-1}$. Moreover, if for any $s \in \tau$, $Q(s)>0$ and no proper postfix of $s$ has the property $(\ref{eqQadapt})$, then $\tau$ is called the minimal context tree of the distribution $Q$, and $(X_{n})_{n\in\bigZ}$ is called a variable length Markov chain (VLMC).
\end{defi}  

If a tree $\tau$ is $Q$-adapted, then for all sequences $x_{-\infty:-1}$ such that  for any $M \ge 1$, $Q(x_{-M:-1})>0$, there exists a unique string in $\tau$ which is postfix of $x_{-\infty:-1}$. We denote this postfix by $\tau(x_{-\infty:-1})$.\\
%
A tree $\tau$ is said to be a \textit{subtree} of $\tau'$ if for each string $s'$ in $\tau'$ there exists a string $s$ in $\tau$ which is postfix of $s'$. Then if $\tau$ is a $Q$-adapted tree, any tree $\tau'$ such that $\tau$ is a subtree of $\tau'$ will be $Q$-adapted.

\begin{defi}\label{defi minimal complete context tree}
Let $Q$ be the distribution of a VLMC $(X_{n})_{n\in\bigZ}$. Let $\tau_{0}$ be its minimal context tree. 
There exists a unique complete tree $\tau^{\star}$ such that
$\tau_{0}$ is a subtree of $\tau^{\star}$
and
\begin{equation*}
|\tau^{\star}| = \min \big\{|\tau|\;:\;\ \tau \text{ is a complete tree and }  \tau_{0}  \text{ is a subtree  of } \tau \big\}.
\end{equation*}
$\tau^{\star}$ is called the minimal complete context tree of the distribution $Q$ of the VLMC $(X_{n})_{n\in\bigZ}$.
\end{defi}

Let us define, for any complete tree $\tau$, the set of transition parameters:
\begin{equation*}
\Theta_{t,\tau} = \bigg\{\left( P_{s,i} \right)_{s \in \tau, i\in\bigX} \;:\; \forall s \in \tau, \forall i \in \bigX, \ P_{s,i} \ge 0 \ \text{and} \ \sum\limits_{i=1}^{k} P_{s,i} = 1     \bigg\}.
\end{equation*}
If $(X_{n})_{n\in\bigZ}$ is a VLMC with  minimal complete context tree $\tau^{\star}$ and transition parameters $\theta_{t}^{\star}=\left( P^{\star}_{s,i} \right)_{s \in \tau^{\star},  i\in\bigX}\in\Theta_{t,\tau^{\star}}$, 
for any complete tree $\tau$ such that $\tau^{\star}$ is a subtree of $ \tau$, there exists a unique $\theta_{t}=\left( P_{s,i} \right)_{s \in \tau,  i\in\bigX}\in\Theta_{t,\tau}$ that defines the same VLMC transition probabilities, namely: for any $s\in\tau$, there exists a unique $u\in\tau^{\star}$ which is a postfix of $s$, and for all $i\in\bigX$, $P_{s,i}= P^{\star}_{u,i}$. 
Of course, a parameter in $\Theta_{t,\tau}$ might  be not sufficient to define a unique distribution of a VLMC (if there is no unique stationary distribution). But the parameter defines a unique distribution of VLMC if, for instance, the Markov chain $([X_{n-d(\tau)+1},\ldots,X_{n}])_{n\in\bigZ}$ it defines is irreducible.

\subsection{Variable length hidden Markov models}

A variable length hidden Markov model (VLHMM) is a bivariate stochastic process $(X_{n},Y_{n})_{n \geq 0}$ where $(X_{n})_{n \geq 0}$
(the state sequence) is a (non observed) stochastic process which is the restriction to non negative indices of a VLMC $(X_{n})_{n \in \bigZ}$ with values in  $\bigX$ and, conditionally on $(X_{n})_{n\geq 0}$, $(Y_{n})_{n \geq 0}$ is a sequence of independent variables in the state space $\bigY$ such that for any integer $n$, the conditional distribution of $Y_{n}$ given the state sequence (called the emission distribution) depends on $X_{n}$ only. 

We assume that the emission distributions are absolutely continuous with respect to some positive measure $\mu$ on $(\bigY,{\cal{F}}_{\bigY})$  and are parametrized by a set of parameters $\Theta_{e} \subset (\bigR^{d_{e}})^{k}\times \bigR^{m_{e}}$, so that the set of emission densities (the possible densities of the distribution of $Y_n$ conditional to $X_{n}=x$) is 
$\{(g_{\theta_{e,x},\eta}(.))_{x\in\bigX}, \ \theta_{e}=(\theta_{e,1},\ldots,\theta_{e,k},\eta)\in \Theta_{e}\}$.
For any  complete tree $\tau$, we define now the parameter set :
$$
\Theta_{\tau} = \Theta_{t,\tau} \times \Theta_{e},
$$
and define, for $\theta = (\theta_{t},\theta_{e})\in \Theta_{\tau}$, $\bigP_{\theta}$ the probability of the VLHMM  $(X_{n},Y_{n})_{n\geq 0}$ such that $(X_{n})_{n \in\bigZ}$ is the VLMC with complete context tree $\tau$, transition parameter $\theta_{t}$, and
for any $(u_{1},u_{2})\in \bigN^{2}$, $u_{1}\le u_{2} $, any sets $A_{u_{1}},\ldots,A_{u_{2}}$ in $\mathcal{F}_{Y}$, any $x_{u_{1}:u_{2}}\in\bigX^{u_{2}-u_{1}+1}$,
\begin{equation*}
\bigP_{\theta}\Bigg(Y_{u_{1}} \in A_{u_{1}}, \ldots ,Y_{u_{2}} \in A_{u_{2}} \bigg\vert X_{u_{1}}=x_{u_{1}}, \ldots ,X_{u_{2}}=x_{u_{2}} \Bigg)= \prod\limits_{u=u_{1}}^{u_{2}}  \left[\int_{A_{u}}g_{\theta_{e,x_{u}},\eta}(y)d\mu(y)\right]. 
\end{equation*}
Of course, as noted before, it can happen that $\theta_{t}$ does not define a unique VLHMM. We shall however do not consider this question since we shall assume that the true parameter defines an irreducible hidden VLMC, and we shall introduce initial distributions to define a computable likelihood: 
throughout the paper we shall assume that the observations $(Y_{1},...,Y_{n})= Y_{1:n}$ come from a VLHMM with parameter $\theta^{\star}$ such that $\tau^{\star}$ is the minimal \textit{complete} context tree of the hidden VLMC, and such that $([X_{n-d(\tau^{\star})+1},\ldots,X_{n}])_{n\in\bigZ}$ is a stationary and irreducible Markov chain. 
And
to define a computable  likelihood, we introduce, for any positive integer $d$, a probability distribution $\nu_{d}$ on $\bigX^{d}$ so that, for any complete tree $\tau$ and any $\theta = (\theta_{t},\theta_{e})\in \Theta_{\tau}$, we set what will be called the likelihood: 
\begin{equation}
\label{eq:lik}
\forall y_{1:n}\in\bigY^{n},  \ g_{\theta}(y_{1:n}) =   \sum\limits_{x_{1:n} \in \bigX^{n}} \left[ \prod\limits_{i=1}^{n} g_{\theta_{e,x_{i}},\eta}(y_{i})\right] g_{\theta_{t}} (x_{1:n})
\end{equation}
where, if $\theta_{t}=\left( P_{s,x} \right)_{s \in \tau, x\in\bigX}$:
\begin{equation}
  \label{vrais}
 g_{\theta_{t}} (x_{1:n})
= \sum\limits_{x_{-d(\tau)+1:0} \in \bigX^{d(\tau)}} \Big[\nu_{d(\tau)}(x_{-d(\tau)+1:0}) \\ 
\prod\limits_{i = 1}^{n} P_{\tau(x_{-d(\tau)+i:i-1}),x_{i}} \Big].
  \end{equation} 
 
 We are concerned with the statistical estimation of the tree $\tau^{\star}$ using a method that involves no prior upper  bound on the depth of $\tau^{\star}$.   
 Define the following estimator of the minimal complete context tree $\tau^{\star}$ : 
   \begin{equation}
    \label{tau} 
   \hat{\tau}_{n} = \argmin\limits_{\tau \ \text{complete  tree}} \Big[ - \sup\limits_{\theta \in \Theta_{\tau}} \log g_{\theta}(Y_{1:n})+ pen(n,\tau)  \Big] 
 \end{equation}  
   where $ pen(n,\tau)$ is a penalty term depending on the number of observations $n$ and the complete tree $\tau$.\\
The label switching phenomenon occurs in statistical inference of VLHMM as it occurs in statistical inference of HMM and of population mixtures. That is:  applying a label permutation on $\bigX$ does not change the distribution of $(Y_{n})_{n\geq 0}$. Thus, if $\sigma$ is a permutation of $\left\{ 1,...,k \right\}$ and $\tau$ is a complete tree, we define the complete tree $\sigma(\tau)$ by 
\[\sigma(\tau) = \left\{ \sigma(x_{1})...\sigma(x_{l}) \big| \ x_{1:l} \in \tau\right\}.
\]  
\begin{defi}\label{defi equivalence tree}
If $\tau$ and $\tau'$ are two complete trees, we say that $\tau$ and $\tau'$ are equivalent, and denote it by $\tau \sim \tau'$, if there exists a permutation $\sigma$ of $\bigX$ such that $\sigma(\tau)= \tau'$.
\end{defi}
We then choose $pen(n,\tau)$ to be invariant by permutation, that is: for any permutation $\sigma$ of $\bigX$, 
$pen(n,\sigma(\tau))=pen(n,\tau)$.
In this case, for any complete tree $\tau$,
 \begin{equation*}
- \sup\limits_{\theta \in \Theta_{\hat{\tau}_{n}}}\log g_{\theta}(Y_{1:n}) + pen(n,\tau)  = \\
- \sup\limits_{\theta \in \Theta_{\sigma(\hat{\tau}_{n})}}\log g_{\theta}(Y_{1:n}) + pen(n,\sigma(\tau))
 \end{equation*}
so that the definition of $\hat{\tau}_{n}$ requires a choice in the set of minimizers of (\ref{tau}).

Our aim is now to find penalties allowing to prove the strong consistency of $\hat{\tau}_{n}$, that is such that $\hat{\tau}_{n}\sim\tau^{\star}$, $\bigP_{\theta^{\star}}$- eventually almost surely as $n \to \infty$.

\section{The general strong consistency theorem}
\label{sec:theo}
In this section, we first recall the tools borrowed from information theory, and set the result that we use in order to find a penalty insuring the strong consistency of $\hat{\tau}_{n}$. 
Then we give our general strong consistency theorem, and straightforward applications. Application to Gaussian emissions with unknown variance, which is more involved, is deferred to the next section.

\subsection{An information theoretic inequality}

We shall introduce mixture probability distributions on $\bigY^{n}$ and compare them to the maximum likelihood, in the same way as  \cite{KT81} first did; see also \cite{Cato} and \cite{EG11} for tutorials and use of such ideas in statistical methods.
For any complete tree $\tau$, we define, for all positive integer $n$, the mixture measure $\bigKT_{\tau}^{n}$ on $\bigY^{n}$ using a prior $\pi^{n}$ on $\Theta_{\tau}$: 
$$ \pi^{n}(d\theta) = \pi_{t}(d\theta_{t})\otimes\pi_{e}^{n}(d\theta_{e})
$$
where $\pi_{e}^{n}$ is a prior on $\Theta_{e}$ that may change with $n$, and $\pi_{t}$ the prior on $\Theta_{t}$ such that, if $\theta_{t}=(P_{s,i})_{s \in \tau,i \in \bigX}$,
$$
\pi_{t}(d\theta_{t}) = \otimes_{s \in \tau} \pi_{s}(d(P_{s,i})_{i \in \bigX}) 
$$
where  $(\pi_{s})_{s\in \tau}$ are Dirichlet $\mathcal{D}(\dfrac{1}{2},...,\dfrac{1}{2})$ distributions on $\left[ 0,1\right]^{|\bigX|}$.
Then $\bigKT_{\tau}^{n}$ is defined on $\bigY^{n}$ by 
$$
\bigKT_{\tau}^{n}(y_{1:n})= \sum\limits_{x_{1:n} \in \bigX^{n}} \bigKT_{\tau,t}(x_{1:n})\bigKT_{e}^{n}(y_{1:n}|x_{1:n}) 
$$
where
$$
 \bigKT_{e}^{n}(y_{1:n}|x_{1:n})=\int\limits_{  \Theta_{e}} \left[ \prod_{i=1}^{n}g_{\theta_{e,x_{i}},\eta}(y_{i})\right] \pi_{e}^{n} (d \theta_{e} )
$$
and
\begin{equation*}
\bigKT_{\tau,t}(x_{1:n})  = 
\left(\frac{1}{k}\right)^{d(\tau)}\int_{\Theta_{t}}\bigP_{\theta_{t}} \left(x_{d(\tau)+1:n} \vert x_{1:d(\tau)} \right)\pi_{t} (d \theta_{t} )
= 
\left(\frac{1}{k}\right)^{d(\tau)}
  \prod\limits_{s\in \tau} \int\limits_{[0,1]^{\left|\bigX \right|} } \prod\limits_{i = 1}^{k} P_{s,i}^{a_{s}^{x}(x_{1:n})} \pi_{s}(d(P_{s,i})_{i \in \bigX})
\end{equation*}
where  $a_{s}^{x}(x_{1:n})$ is the number of times that $x$ appears in context $s$, that is $a_{s}^{x}(x_{1:n})=\sum_{i=d(\tau)+1}^{n}\one_{x_{i}=x, x_{i-l(s),i-1}=s}$.

The following inequality will be a key tool to control the fluctuations of the likelihood.
\begin{pr}
\label{pr ineg KT}
There exists a finite constant $D$ depending only on $k$ such that for any complete tree $\tau$, and any  $y_{1:n}\in \bigY^{n}$: 
\begin{equation*}
0\le\sup\limits_{\theta \in \Theta_{\tau}}\log g_{\theta}(y_{1:n})-\log\bigKT_{\tau}^{n} (y_{1:n}) \le 
\sup\limits_{x_{1:n}}  \Bigg[  \log  \prod_{i=1}^{n}g_{\theta_{e,x_{i}},\eta}(y_{i}) - \log \bigKT_{e}^{n}(y_{1:n}|x_{1:n})\Bigg] +\dfrac{k-1}{2}|\tau|\log n + D
\end{equation*}
\end{pr}
\begin{proof}
Let $\tau$ be a complete tree. For any $\theta\in\Theta_{\tau}$,
\begin{eqnarray*}
\dfrac{g_{\theta}(y_{1:n})}{\bigKT_{\tau}^{n} (y_{1:n})} & = & \dfrac{ \sum\limits_{x_{1:n}}  g_{\theta_{t}} (x_{1:n}) \prod_{i=1}^{n}g_{\theta_{e,x_{i}},\eta}(y_{i})}{ \sum\limits_{x_{1:n}}  \bigKT_{\tau} (x_{1:n})\bigKT_{e}^{n}(y_{1:n}|x_{1:n})} \\
& \leq &  \max\limits_{x_{1:n}} \dfrac{ g_{\theta_{t}}(x_{1:n})  \prod_{i=1}^{n}g_{\theta_{e,x_{i}},\eta}(y_{i})}{  \bigKT_{\tau} (x_{1:n})\bigKT_{e}^{n}(y_{1:n}|x_{1:n})}.
\end{eqnarray*}
Thus, 
\begin{equation*}
 \log \dfrac{g_{\theta}(y_{1:n})}{\bigKT_{\tau}^{n} (y_{1:n})}   \leq   \sup\limits_{x_{1:n}} \Bigg[  \log \prod_{i=1}^{n}g_{\theta_{e,x_{i}},\eta}(y_{i})\\
  - \log \bigKT_{e}^{n}(y_{1:n}|x_{1:n})  + |\tau| \gamma \bigg(\dfrac{n}{|\tau|} \bigg)+ d(\tau) \log k  \Bigg]
\end{equation*}
where $\gamma(x) = \dfrac{k-1}{2} \log x + \log k$, using \cite{EG11}.
Then
\begin{equation*} \log \dfrac{g_{\theta}(y_{1:n})}{\bigKT_{\tau}^{n} (y_{1:n})}\leq \sup\limits_{x_{1:n}}  \Bigg[  \log \prod_{i=1}^{n}g_{\theta_{e,x_{i}},\eta}(y_{i})\\ - \log \bigKT_{e}^{n}(y_{1:n}|x_{1:n})\Bigg] +  \dfrac{k-1}{2}|\tau|\log n + D(\tau)
\end{equation*}
where $D(\tau)= - \dfrac{k-1}{2}|\tau|\log |\tau|+ |\tau| \log k +d(\tau)\log k$.  Now, since $\tau$ is complete, $d(\tau) \leq \dfrac{|\tau|-k}{k-1}$, so that
$$
D(\tau)\le |\tau| \Big( \log k -\dfrac{k-1}{2}\log |\tau|\Big) + \dfrac{|\tau|-k}{k-1} \log k.
$$
But the upper bound in the inequality  tends to $-\infty$ when $|\tau|$ tends to $\infty$, so that  there exists a constant $D$ depending only on $k$ such that for any complete tree $\tau$, $D(\tau) \le D$.
\end{proof}

\subsection{Strong consistency theorem}

Let $\theta^{\star} = (\theta_{t}^{\star},\theta_{e}^{\star})$ with $\theta_{t}^{\star}= (P^{\star}_{s,i})_{s \in \tau^{\star}, i\in\bigX}$, and $\theta_{e}^{\star}= (\theta_{e,1}^{\star},...,\theta_{e,k}^{\star},\eta^{\star})$ be the true parameters of the VLHMM.\\
Let us now define for any positive $\alpha$, the penalty:
\begin{equation}
\label{pen} 
\quad pen_{\alpha}(n,\tau) = \left[\sum\limits_{t=1}^{|\tau|} \dfrac{(k-1)t + \alpha}{2}\right] \log n
\end{equation}
Notice that the complexity of the model is taken into account through the cardinality of the tree $\tau$.\\
We need to introduce further assumptions. 
\begin{itemize}
\item
\textbf{(A1)}. The Markov chain $(({X}_{n-d(\tau^{\star})+1},\ldots, X_{n}))_{n\geq d(\tau^{\star})}$ is irreducible.
\item
\textbf{(A2)}. For any complete tree $\tau$ such that $|\tau| \leq |\tau^{\star}|$ and which is not equivalent to $\tau^{\star}$, for any $\theta\in\Theta_{\tau}$, the random sequence $(\theta_{e,X_{n}})_{n\in \bigZ}$ where $(X_{n})_{n \in \bigZ}$ is a VLMC with transition probabilities $\theta_{t}$, has a different distribution than  $(\theta^{\star}_{e,X_{n}})_{n\in \bigZ}$ where $(X_{n})_{n \in \bigZ}$ is a VLMC with transition probabilities $\theta^{\star}_{t}$.
\item
\textbf{(A3)}. The family $\left\{g_{\theta_{e}},\theta_{e} \in \Theta_{e},\right\}$ is such that for any probability distributions $(\alpha_{i})_{i=1,...,k}$ and $(\alpha'_{i})_{i=1,...,k}$ on 
$\left\{ 1,...,k\right\}$, any $(\theta_{1},\ldots,\theta_{k},\eta)\in \Theta_{e}$ and $(\theta'_{1},\ldots,\theta'_{k},\eta')\in \Theta_{e}$, if 
$$
\sum\limits_{i=1}^{k} \alpha_{i} g_{\theta_{i},\eta} =\sum\limits_{i=1}^{k} \alpha'_{i} g_{\theta'_{i},\eta'}$$ then, \\ $$\sum\limits_{i=1}^{k} \alpha_{i} \delta_{\theta_{i}} = \sum\limits_{i=1}^{k} \alpha'_{i} \delta_{\theta'_{i}} \ \text{and} \ \eta=\eta' $$

\item
\textbf{(A4)}. For any $y\in\bigY$, $\theta_{e} \longmapsto g_{\theta_{e}}(y) = (g_{\theta_{e,i}, \eta}(y))_{i\in \bigX}$ is continuous
 and tends to zero when $||\theta_{e}||$ tends to infinity.
 \item
\textbf{(A5)}. For any $i\in\bigX$, $E_{\theta^{\star}}\left[ |\log g_{\theta^{\star}_{e,i},\eta^{\star}}(Y_{1})|\right]<\infty$.
\item
\textbf{(A6)}. For any $\theta_{e} \in \Theta_{e}$,
there exists $\delta>0$ such that : $E_{\theta^{\star}}\left[ \sup\limits_{ ||\theta_{e}' -\theta_{e} ||< \delta }(\log g_{\theta_{e}'}(Y_{1}))^{+}\right]<\infty$ .
\end{itemize}

\begin{theo}\label{theo strog consistency}
Assume that {\bf{(A1)}} to {\bf{(A6)}} hold, and that moreover there exists
a positive real number $b$ such that
\begin{equation}
\label{hypoprior}
\sup\limits_{\theta_{e} \in \Theta_{e}}\sup\limits_{x_{1:n}}  \Bigg[  \log  \prod_{i=1}^{n}g_{\theta_{e,x_{i}},\eta}(Y_{i}) - \log \bigKT_{e}^{n}(Y_{1:n}|x_{1:n})\Bigg] \le   b\log n \quad 
\end{equation}
$\bigP_{\theta^{\star}}$ - eventually almost surely.
If one chooses $\alpha>2(b+1)$ in the penalty (\ref{pen}), then $\hat{\tau}_{n}\sim\tau^{\star}$, $\bigP_{\theta^{\star}}$ - eventually almost surely.
\end{theo}
Notice that, to apply this theorem, one has to find a sequence of priors $\pi_{e}^{n}$ on $\Theta_{e}$ such that (\ref{hypoprior}) holds. The remaining of the section will prove that it is possible for situations in which priors may be defined as in previous works about HMM order estimation, while in the next section, we will prove that it is possible to find a prior in the important case of Gaussian emissions with unknown variance.\\
In the following proof, the assumption (\ref{hypoprior}) insures that $|\hat{\tau}_{n}| \le |\tau^{\star}|$  eventually almost surely, while assumptions {\bf(A1-6)} insure that for any complete tree $\tau$ such that $|\tau| < |\tau^{\star}|$ or  $|\tau| = |\tau^{\star}|$  and $\tau \nsim  \tau^{\star}$, $\hat{\tau}_{n} \neq \tau^{\star}$ $\bigP_{\theta^{\star}}$ - eventually almost surely. In particular { \bf(A2)} holds whenever $\theta_{e,x}^{\star} \neq \theta_{e,y}^{\star}$ if $(x,y) \in \bigX^{2}$ and $x \neq y$. \\
\begin{proof}
The proof will be structured as follow : we first prove that  $\bigP_{\theta^{\star}}$ - eventually almost surely, $|\hat{\tau}_{n}| \le |\tau^{\star}|$. We then prove that for any complete tree $\tau$ such that $|\tau|\le|\tau^{\star}|$ and $\tau \nsim \tau^{\star}$, $\hat{\tau}_{n} \nsim \tau$ $\bigP_{\theta^{\star}}$ - eventually almost surely. This will end the proof since there is a finite number of such trees. For any $n\in\bigN$, we denote by $E_{n}$ the event \\[0.2cm]
\begin{equation*}
E_{n} \ : \ \Bigg[
\sup\limits_{\theta_{e} \in \Theta_{e}}\sup\limits_{x_{1:n}}  \bigg(  \log  \prod_{i=1}^{n}g_{\theta_{e,x_{i}},\eta}(Y_{i})
 - \log \bigKT_{e}^{n}(Y_{1:n}|x_{1:n})\bigg) \le b \log n \quad 
 \Bigg]
 \end{equation*} 
By using (\ref{hypoprior}) and Borel-Cantelli Lemma, to get that $\bigP_{\theta^{\star}}$ - eventually almost surely, $|\hat{\tau}_{n}| \le |\tau^{\star}|$,
it is enough to show that \[ \sum\limits_{n=1}^{\infty} \bigP_{\theta^{\star}} \left\{\left( \left| \hat{\tau}_{n}\right| >  \left|\tau^{\star}\right| \right) \bigcap E_{n} \right\} < \infty.
\]
Let $\tau$ be a complete tree such that  $|\tau|>|\tau^{\star}|$. Using Proposition \ref{pr ineg KT},
\begin{eqnarray*}
 \bigP_{\theta^{\star}} \left\{\left(  \hat{\tau}_{n}= \tau^{\star} \right) \bigcap E_{n} \right\} & \le &  \bigP_{\theta^{\star}} \Bigg\{\bigg(   \sup\limits_{\theta \in \Theta_{\tau}} \log g_{\theta}(Y_{1:n}) - pen_{\alpha}(n,\tau) \geq   \log g_{\theta^{\star}}(Y_{1:n}) -  pen_{\alpha}(n,\tau^{\star})\bigg)\bigcap E_{n}\Bigg\}\\
&\le& \bigP_{\theta^{\star}} \Bigg\{\bigg(   \log \bigKT_{\tau}^{n}(y_{1:n}) +  \sup\limits_{\theta_{e} \in \Theta_{e}} \sup\limits_{x_{1:n}}  \big[  \log  \prod_{i=1}^{n}g_{\theta_{e,x_{i}},\eta}(Y_{i})  - \log \bigKT_{e}^{n}(Y_{1:n}|x_{1:n})\big]\\
& &  \quad \quad +   \dfrac{k-1}{2}|\tau|\log n + D -  \log g_{\theta^{\star}}(Y_{1:n})  +  pen_{\alpha}(n,\tau^{\star})  -  pen_{\alpha}(n,\tau)\geq 0 \bigg) \bigcap E_{n}\Bigg\}  \\[0.4cm]
 & \le &\bigP_{\theta^{\star}}\left\{   g_{\theta^{\star}}(Y_{1:n})\leq  \bigKT_{\tau}^{n}(Y_{1:n})   \right\} \exp\left(e_{\tau,n}\right)
\end{eqnarray*}
with
\begin{equation*}
e_{\tau,n,} = \dfrac{k-1}{2} |\tau| \log n +b\log n +D+ pen_{\alpha}(n,\tau^{\star}) - pen_{\alpha}(n,\tau).
\end{equation*}
But
\begin{eqnarray*}
 e_{\tau,n}  &  = & \dfrac{k-1}{2} |\tau| \log n  + b \log n 
 +D  + \sum\limits_{t=1}^{|\tau^{\star}|} \dfrac{(k-1)t + \alpha}{2} \log n -  \sum\limits_{t=1}^{|\tau|} \dfrac{(k-1)t + \alpha}{2} \log n \\
& = & \dfrac{k-1}{2} |\tau| \log n  +  b \log n+D  - \sum\limits_{t=|\tau^{\star}|+1}^{|\tau|} \dfrac{(k-1)t + \alpha}{2} \log n\\
& \le & -\dfrac{\alpha}{2} \bigg( |\tau| - |\tau^{\star}|\bigg)\log n + b \log n+D,
\end{eqnarray*}
so that
\begin{eqnarray*}
\bigP_{\theta^{\star}} \bigg\{\big(  \hat{\tau}_{n}= \tau^{\star} \big) \bigcap  E_{n} \bigg\} &\le& e^{ -\dfrac{\alpha}{2} \bigg( |\tau| - |\tau^{\star}|\bigg)\log n +b \log n +D}\\
 & = &  C.n^{-\dfrac{\alpha}{2}(|\tau|-|\tau^{\star}|)+b} 
\end{eqnarray*}
 for some constant $C$. Thus
\begin{equation*}
  \bigP_{\theta^{\star}} \left\{\left( \left| \hat{\tau}_{n}\right| >  \left|\tau^{\star}\right| \right) \bigcap E_{n} \right\}   \leq C  \sum\limits_{t=|\tau^{\star}|+1}^{\infty} CT(t) n^{-\dfrac{\alpha}{2}(t-|\tau^{\star}|)+b}
 \end{equation*}
where $CT(t)$ is the number of complete trees with $t$ leaves. But using Lemma 2 in \cite{Gari06}, $CT(t)\leq 16^{t}$ so that
\begin{eqnarray*}
  \bigP_{\theta^{\star}} \left\{\left( \left| \hat{\tau}_{n}\right| >  \left|\tau^{\star}\right| \right) \bigcap E_{n} \right\} & \leq &  Cn^{b} 16^{|\tau^{\star}|} \sum\limits_{t=1}^{\infty} \big[ 16n^{-\alpha/2}\big]^{t}\\
&= & O(n^{-\alpha/2+b})
\end{eqnarray*}
which is summable  if  $\alpha>2(b+1)$.\\

Let now $\tau$ be a tree such that $|\tau| \le |\tau^{\star}|$ and $\tau \nsim \tau^{\star}$.
Let $\tau_{M}$ be a complete tree such that $\tau$ and $\tau^{\star}$ are both a subtree of $\tau_{M}$. 
Then, by setting for any integer $n\geq d(\tau_{M})-1$, $W_{n} = [X_{n-d(\tau_{M})+1 :n}]$, for any $\theta\in\Theta_{\tau}\cup\Theta_{\tau^{\star}}$,
$(W_{n},Y_{n})_{n \in \bigZ}$ is a HMM under $\bigP_{\theta}$.
Following the proof of  Theorem 3 of  \cite{L92}, we obtain that
there exists $K>0$ such that $\bigP_{\theta^{\star}}$-eventually a.s., 
\[\dfrac{1}{n} \log g_{\theta^{\star}}(Y_{1:n}) - \sup\limits_{\theta \in \Theta_{\tau}}\dfrac{1}{n} \log g_{\theta}(Y_{1:n}) \ge K\]
so that 
\begin{equation*}
\log g_{\theta^{\star}}(Y_{1:n})- pen(n,\tau^{\star}) - \sup\limits_{\theta \in \Theta_{\tau}} \log g_{\theta}(Y_{1:n}) +pen(n,\tau) >0, \bigP_{\theta^{\star}}\text{-eventually a.s.,}
\end{equation*}
which finishes the proof of Theorem \ref{theo strog consistency}.
\end{proof}

\subsection{Gaussian emissions with known variance} 
Here, we do not need the parameter $\eta$  so we omit it. Then $\Theta_{e}=\{\theta_{e}=(m_{1},\ldots,m_{k})\in\bigR^{k}\}$. The conditional likelihood is given, for any $\theta_{e}=(m_{x})_{x\in\bigX}$ by
\[g_{\theta_{e,x}}(y) = \dfrac{1}{\sqrt{2\pi\sigma^{2}}} \exp(-\dfrac{(y -m_{x})^{2}}{2\sigma^{2}}).\]
\begin{pr}
\label{pr-gaus}
Assume {\bf (A1-2)}.
If one chooses
$\alpha > k+2$ in the penalty (\ref{pen}),  $ \hat{\tau}_{n} \sim \tau^{\star}$, $\bigP_{\theta^{\star}}$ - eventually a.s.
\end{pr}
\begin{proof}

The identifiability of the Gaussian model {\bf (A3)} has been proved by Yakowitz and Spragins in \cite{yakowitz}, it is easy to see that Assumptions {\bf (A4)} to {\bf (A6)} hold. 
Now, we define the prior measure $\pi_{e}^{n}$ on $\Theta_{e}$ as the probability distribution under which $\theta_{e}=(m_{1}, . . . ,m_{k})$ is a vector of $k$ independent random variables with centered Gaussian distribution with variance $\tau^{2}_{n}$. Then, using \cite{CGG09}, $\bigP_{\theta^{\star}}$ -eventually a.s., 

\begin{equation*}
\sup\limits_{\theta_{e}\in \Theta_{e}} \max\limits_{x_{1:n} \in \bigX^{n}} \bigg[\log \prod_{i=1}^{n}g_{\theta_{e,x_{i}}}(Y_{i})- \log  \bigKT_{e}^{n}(Y_{1:n}|x_{1:n}) \bigg]\le \dfrac{k}{2} \log(1 + \dfrac{n\tau^{2}_{n}}{k\sigma^{2}}) + \dfrac{k}{2\tau^{2}_{n}} 5 \sigma^{2} \log n.
\end{equation*}
Thus, by choosing $\tau_{n}^{2} = \dfrac{5\sigma^{2}k\log(n)}{2}$, we get that for any $\epsilon >0$, 

\begin{equation*}\sup\limits_{\theta_{e} \in \Theta_{e}} \max\limits_{x_{1:n} \in \bigX^{n}} \bigg[\log \prod_{i=1}^{n}g_{\theta_{e,x_{i}}}(Y_{i}) - \log  \bigKT_{e}^{n}(Y_{1:n}|x_{1:n}) \bigg] \le \dfrac{k+\epsilon}{2} \log n 
\end{equation*} 

$\bigP_{\theta^{\star}}$ -eventually almost surely, and (\ref{hypoprior}) holds for any $b > \dfrac{k}{2}$.
\end{proof}

\subsection{Poisson emissions}
Now the conditional distribution of $Y$ given $X=x$ is Poisson with mean $m_{x}$ and $\Theta_{e}=\left\{\theta_ {e}=(m_{1},...,m_{k}) \big| \ \forall j\in \bigX, \ m_{j}>0\right\}$.
\begin{pr}
Assume {\bf (A1-2)}.
If one chooses $\alpha > k+2$ in  in the penalty (\ref{pen}), $ \hat{\tau}_{n} \sim \tau^{\star}$ $\bigP$ -eventually a.s.
\end{pr}
\begin{proof}

The identifiability of the Gaussian model {\bf (A3)} has been proved by Teicher in \cite{Teicher}, it is easy to see that Assumptions {\bf (A4)} to {\bf (A6)} hold.
The prior $\pi_{e}^{n}$ on $\Theta_{e}$ is now defined such that
$m_{1},...,m_{k}$ are independent identically distributed with distribution Gamma($t$, 1/2). Then,
using \cite{CGG09}: 
\begin{equation*}\sup\limits_{\theta_{e}\in \Theta_{e}} \max\limits_{x_{1:n} \in \bigX^{n}} \bigg\{\log \prod_{i=1}^{n}g_{\theta_{e,x_{i}}}(Y_{i})- \log  \bigKT_{e}^{n}(Y_{1:n}|x_{1:n})\bigg\} \le \dfrac{k}{2} \log \dfrac{n}{k} 
+ k t \dfrac{\log n}{\sqrt{\log \log n}} + \dfrac{k}{2} (1+t\log t)
 \end{equation*}
$\bigP_{\theta^{\star}}$ -eventually a.s..
Then, for any fixed $t>0$, for any $\epsilon >0$, eventually almost surely :

\begin{equation*}
\sup\limits_{\theta_{e}=(m_{1},...,m_{k}) \in \Theta_{e}} \max\limits_{x_{1:n} \in \bigX^{n}} \bigg\{ \log g_{\theta_{e}}(Y_{1:n}|x_{1:n}) - \log  \bigKT_{e}^{n}(Y_{1:n}|x_{1:n})\bigg\}  \le \left( \dfrac{k}{2} + \epsilon\right) \log n
\end{equation*}
$\bigP_{\theta^{\star}}$ -eventually almost surely, and (\ref{hypoprior}) holds for any $b > \dfrac{k}{2}$.
\end{proof}

\section{Gaussian emissions with unknown variance }
\label{sec:gauss}

We consider the situation where the emission distributions are Gaussian with the same, but unknown, variance $\sigma_{\star}^{2}$ and with a mean depending on the hidden state $x$. Let $\eta=-\dfrac{1}{2\sigma^{2}}$ and $ \theta_{e,j} = \dfrac{m_{j}}{\sigma^{2}}$ for all $j \in \bigX=\left\{1,..,k\right\}$. Here $$\Theta_{e}=\left\{ \left(\eta,\big(\theta_{e,j}\big)_{j=1,...,k}\right)   \ \bigg| \ \theta_{e,j}\in\bigR,\eta<0 \right\} \ .$$
If $x_{1:n}\in \bigX^{n}$, for any $j\in \bigX$, we set $I_{j} = \left\{i | \\ x_{i}=j \right\}$ and $n_{j} = |I_{j}|$. 
For sake of simplicity we omit  $x_{1:n}$ in the notation though $I_{j}$ and $n_{j}$ depend on $x_{1:n}$.
The conditional likelihood is given, for any $x_{1:n}$ in $\bigX^{n}$, for any $y_{1:n}$ in $\bigY^{n}$,  by
\begin{equation*}
\prod_{i=1}^{n}g_{\theta_{e,x_{i}},\eta}\left(y_{i}\right)= \dfrac{1}{\sqrt{2\pi}^{n}} \prod\limits_{j=1}^{k} \exp\bigg[ \eta \sum\limits_{i\in I_{j}} y_{i}^{2} + \theta_{e,j}\sum\limits_{i\in I_{j}} y_{i} - n_{j} A\left( \eta,\theta_{e,j}\right)\bigg] 
\end{equation*}
where \[ A\left( \eta,\theta_{e,j}\right) = -\dfrac{\theta_{e,j}^{2}}{4\eta } - \dfrac{1}{2} \log(-2\eta) \]
\begin{theo}
\label{theorem gauss}
Assume {\bf (A1-2)}.
If one chooses $\alpha> k+3$ in the penalty (\ref{pen}), then $\hat{\tau}_{n} \sim \tau^{\star}$, $\bigP_{\theta^{\star}}$ - eventually a.s.
\end{theo}
\begin{proof}
We shall prove that Theorem \ref{theo strog consistency} applies. First, it is easy to see that Assumptions {\bf (A4)} to {\bf (A6)} hold and the proof of {\bf (A3)} can be found in \cite{yakowitz}.\\ 
Define now the  conjugate exponential prior on $\Theta_{e}$ : 
\begin{equation*}
\pi_{e}^{n}(d\theta_{e}) = \exp\bigg[ \alpha_{1}^{n} \eta + \sum\limits_{j=1}^{k} \alpha_{2,j}^{n} \theta_{e,j}
 - \sum\limits_{j=1}^{k} \beta_{j} ^{n} A\left( \eta,\theta_{e,j}\right) - B\left( \alpha_{1}^{n},\alpha_{2,1}^{n},\ldots, \alpha_{2,k}^{n},\beta_{1}^{n},\ldots, \beta_{k}^{n}  \right)\bigg] 
 d\eta d\theta_{e,1}\cdots  d\theta_{e,k}
\end{equation*}
where the parameters $\alpha_{1}^{n}, \ (\alpha_{2,j}^{n})_{j=1,...,k}$ and $(\beta_{j}^{n})_{j=1,...,k}$ will be chosen later, and the normalizing constant may be computed as
\begin{equation*}
\exp \left\{ B\left( \alpha_{1}^{n},\alpha_{2,1}^{n},\ldots, \alpha_{2,k}^{n},\beta_{1}^{n},\ldots, \beta_{k}^{n}   \right)\right\} = \dfrac{2^{k+\frac{\sum_{j=1}^{k} \beta_{j}^{n}}{2}} \pi^{\frac{k}{2}} \Gamma\left( \frac{\sum_{j=1}^{k} \beta_{j}^{n} + k + 2}{2} \right)}{\left(\prod_{j=1}^{k} \sqrt{\beta_{j}^{n}}\right) \left( \alpha_{1} ^{n}- \sum_{j=1}^{k}  \frac{(\alpha_{2,j}^{n})^{2}}{\beta_{j}^{n}} \right)^{\frac{\sum_{j=1}^{k} \beta_{j}^{n} + k + 2}{2}}}
\end{equation*}
where we recall the Gamma function: $\Gamma (z) = \int_{0}^{+\infty} u^{z-1} e ^{-u} du$ for any complex number $z$.
Theorem \ref{theorem gauss} follows now from Theorem \ref{theo strog consistency} and the proposition below.
\end{proof}
\begin{pr}\label{propo KT gaussian unnown var}
 If {\bf{(A1)}} holds, 
it is possible to choose the parameters $\alpha_{1}^{n}, \ (\alpha_{2,j}^{n})_{j=1,...,k}$ and $(\beta_{j}^{n})_{j=1,...,k}$
such that for any $\epsilon>0$,
\begin{equation*}
\max\limits_{x_{1:n}}\bigg\{\sup_{\theta_{e}\in\Theta_{e}}\prod_{i=1}^{n}g_{\theta_{e,x_{i}},\eta}\left(Y_{i}\right) - \log  \bigKT_{e}^{n}(Y_{1:n}|x_{1:n})\bigg\} \le \dfrac{k+1+\epsilon}{2} \log n 
\end{equation*}
 $\bigP_{\theta^{\star}}$ - eventually a.s.
\end{pr}
\begin{proof}
For any $x_{1:n}\in\bigX^{n}$,
the parameters $\left(\widehat{ \eta},(\widehat{\theta}_{e,j})_{j}\right)$ maximizing the conditional likelihood are given by 
$$\widehat{\eta} = -\frac{1}{2\widehat{\sigma}_{x_{1:n}}^{2}},\;
\widehat{\theta}_{e,j} = \frac{\widehat{m}_{x_{1:n},j}} {\widehat{\sigma}_{x_{1:n}}^{2}}
$$
with
$$
\widehat{m}_{x_{1:n},j}= \frac{\sum_{i\in I_{j}} Y_{i}}{n_{j}},\;
\widehat{\sigma}_{x_{1:n}}^{2}  = \frac{1}{n}\sum_{j=1}^{k} \sum_{i \in I_{j}} \left(Y_{i}- \hat{m}_{x_{1:n},j} \right)^{2} 
$$
so that
$$
\log\prod_{i=1}^{n}g_{\theta_{e,x_{i}},\eta}\left(Y_{i}\right) \leq -n\log \widehat{\sigma}_{x_{1:n}}-\frac{n}{2}\log 2\pi - \frac{n}{2}
.
$$
Also, 
\begin{equation*}
\bigKT_{e}^{n}(y_{1:n}|x_{1:n})= \frac{1}{\sqrt{2\pi}^{n}} \exp \Bigg[B\bigg(\alpha_{1}^{n}+\sum\limits_{i=1}^{n} Y_{i}^{2},
 (\alpha_{2,j}^{n}+\sum\limits_{i \in I_{j}} Y_{i})_{1\leq j \leq k} ,(\beta_{j}^{n}+n_{j})_{1\leq j=1\leq k}  \bigg)
-B\left(\alpha_{1}^{n},(\alpha_{2,j}^{n})_{1\leq j \leq k} ,(\beta_{j}^{n})_{1\leq j \leq k}  \right) \Bigg].
\end{equation*}
Recall that  for all $z>0$ (see for instance \cite{WhiWa}) 
$$
\sqrt{2 \pi} e^{-z} z^{z-\frac{1}{2}} \le \Gamma (z) \le \sqrt{2 \pi} e^{-z + \frac{1}{12z} } z^{z-\frac{1}{2}}  
$$
so that one gets that, for any $x_{1:n}\in\bigX^{n}$ and any $\theta_{e}\in\Theta_{e}$,
\begin{eqnarray*}
\log\prod_{i=1}^{n} g_{\theta_{e,x_{i}},\eta} \left(Y_{i}\right) - \log \bigKT_{e}^{n}(y_{1:n}|x_{1:n}) &\leq& o(\log n) -\frac{n}{2} \log  \widehat{\sigma}_{x_{1:n}}^{2}  - \frac{n}{2}\left(1 + \log 2\right)+   \frac{k}{2}\log \left( \frac{n + \sum_{j = 1}^{k} \beta_{j} ^{n}}{k}\right) \\ 
& & -\bigg[-\frac{n+\sum_{j = 1}^{k} \beta_{j}^{n}+k+2}{2}\\
& & \quad \quad +\left( \frac{n + \sum_{j = 1}^{k} \beta_{j}^{n} +k+1}{2}\right)\log \frac{n+\sum_{j = 1}^{k} \beta_{j}^{n}+k+2}{2}\bigg]\\
& & + \frac{ n + \sum_{j = 1}^{k} \beta_{j}^{n} + k+2}{2} \log \Bigg(  \alpha_{1}^{n} + \sum_{i=1}^{n} Y_{i}^{2}
-  \sum\limits_{j=1}^{k} \frac{\left( \alpha_{2,j}^{n} + \sum_{i\in I_{j}} Y_{i}\right)^{2}}{n_{j}+\beta_{j}^{n}} \Bigg)
\end{eqnarray*}
Choose now
\begin{equation}
\label{equapara}
\beta_{j}^{n}=\alpha_{2,j}^{n}=\frac{1}{n},\;\alpha_{2,j}^{n}=\sqrt{\beta_{j}^{n}},
\;j=1,\ldots,k,\;
\alpha_{1}^{n}=k+1.
\end{equation}
Then one easily gets that for any $x_{1:n}\in\bigX^{n}$ and any $\theta_{e}\in\Theta_{e}$,
\begin{eqnarray*}
\log\prod_{i=1}^{n}g_{\theta_{e,x_{i}},\eta}\left(Y_{i}\right) & - &\log \bigKT_{e}^{n}(Y_{1:n}|x_{1:n}) \leq o(\log n)\\ &+&\frac{n + \sum_{j = 1}^{k} \beta_{j}^{n} + k + 2}{2} \log \Bigg(1 + \frac{1}{n\widehat{\sigma}_{x_{1:n}}^{2}} \bigg[ k + 1   \quad + \sum_{j=1}^{k} \bigg\{\widehat{m}_{x_{1:n},j}^{2}\left(n_{j}- \frac{n_{j}^{2}}{n_{j}+ 1/n} \right)\\
& &  -2\frac{n_{j}}{n.n_{j}+1}\widehat{m}_{x_{1:n},j} -\frac{1}{n^{2} n_{j} + n} \bigg\} \bigg] \Bigg)+\frac{k+1}{2} \log n + \frac{k/n + k + 2 }{2} \log \widehat{\sigma}_{x_{1:n}}^{2}\\
\end{eqnarray*}
Let now $|Y|_{(n)}=\max_{1\leq i \leq n} |Y_{i}|$. Then  for any $x_{1:n}\in\bigX^{n}$,
$$
 \widehat{\sigma}_{x_{1:n}}^{2} \leq |Y|_{(n)}^{2}\;{\text{and}}\;|\widehat{m}_{x_{1:n},j}| \leq |Y|_{(n)},\;j=1,\ldots ,k.
$$
Also,
for any  partition $(I_{i},\ldots,I_{k})$ of $\bigR$ in $k$ intervals, define :  
 $$
\widehat{ \sigma}_{I_{i},\ldots,I_{k}}^{2}=  \frac{1}{n}\sum\limits_{j=1}^{k} \sum\limits_{i=1}^{n} \one_{Y_{i}\in I_{j}} \left(Y_{i} - \frac{\sum\limits_{i'=1}^{n} \one_{Y_{i'} \in I_{j}} Y_{i'}}{\sum\limits_{i'=1}^{n} \one_{Y_{i'} \in I_{j}}} \right)^{2}
 $$
and  
\[\sigma_{I_{i},\ldots,I_{k}}^{2}=  \sum\limits_{j=1}^{k} \bigP_{\theta^{\star}}(Y_{1} \in I_{k}) Var_{\theta^{\star}}(Y_{1}|Y_{1} \in I_{k})  
 \] 
 where $ Var_{\theta^{\star}}(Y_{1}|Y_{1} \in I_{k}) $ is the conditional variance of $Y_{1}$ given that $Y_{1} \in I_{k}$.
 The \textit{k-means} algorithm, see \cite{McQueen}, \cite{inaba}, allows to find a local minimum of the function $ x_{1:n} \longrightarrow \widehat{\sigma}_{x_{1:n}}^{2}$ starting with any initial configuration $x_{1:n}$. Each step of the algorithm produces an assignment of the values $Y_{1:n}$ in $k$ clusters (by partitioning the observations according to the Vorono\"{i} diagram generated by the means of each cluster). Here, the values $Y_{1:n}$ being real numbers, a Vorono\"{i} diagram clustering on $\bigR$ is nothing else than a clustering by intervals. Because the \textit{k-means} algorithm converges, in a finite time, to a local minimum of the quantity $ x_{1:n} \longrightarrow \widehat{\sigma}_{x_{1:n}}^{2}$, if the initial configuration is the $x^{0}_{1:n}$ that minimizes $\widehat{\sigma}_{x_{1:n}}^{2}$, the $k-means$ algorithm will lead to the same configuration  $x^{0}_{1:n}$. Thus, the minimum of  $\widehat{\sigma}_{x_{1:n}}^{2}$ is  a clustering by intervals, that is
 $$
 \inf_{x_{1:n}\in\bigX^{n}}\widehat{\sigma}_{x_{1:n}}^{2}=\inf_{I_{i},\ldots,I_{k}}\widehat{ \sigma}_{I_{i},\ldots,I_{k}}^{2}
 $$
  where the infimum is over  all   partitions  of $\bigR$  in $k$ intervals.\\
  We now get:
\begin{eqnarray*}
& &\log\prod_{i=1}^{n}g_{\theta_{e,x_{i}},\eta} \left(Y_{i}\right)- \log \bigKT_{e}^{n}(Y_{1:n}|x_{1:n}) \leq o(\log n)\\
 & &+\frac{n + \sum\limits_{j = 1}^{k} \beta_{j}^{n} + k + 2}{2} \log \Bigg(1 + \frac{1}{n  \inf\limits_{I_{i},\ldots,I_{k}}\widehat{ \sigma}_{I_{i},\ldots,I_{k}}^{2}}\bigg[ k + 1 + \sum\limits_{j=1}^{k} \Big[|Y|_{(n)}^{2}\left(n_{j}- \frac{n_{j}^{2}}{n_{j}+ 1/n} \right) +2\frac{n_{j}}{n.n_{j}+1}|Y|{(n)} \Big] \bigg] \Bigg)\\
 & & +\frac{k+1}{2} \log n + \frac{k/n + k + 2 }{2} \log |Y|_{(n)}^{2}\\
\end{eqnarray*}
 and 
 Proposition \ref{propo KT gaussian unnown var} follows from the choice (\ref{equapara}) and the lemmas below, whose proofs are given  in the Appendix.
 \end{proof}
 \begin{lem}
 \label{lemsinf}
 If {\bf{(A1)}} holds,\\ 
 $
 \sup_{I_{i},\ldots,I_{k}} \left \vert  \widehat{ \sigma}_{I_{i},\ldots,I_{k}}^{2}-\sigma_{I_{i},\ldots,I_{k}}^{2} \right \vert
 $
 converges to $0$ as $n$ tends to infinity  $\bigP_{\theta^{\star}}$ - a.s. (Here the supremum is over  all   partitions  of $\bigR$  in $k$ intervals).
Also,
 the infimum $s_{\inf}$ of $\sigma_{I_{i},\ldots,I_{k}}^{2}$ over all   partitions  of $\bigR$  in $k$ intervals satisfies $s_{\inf}>0$.
 \end{lem}
 \begin{lem}
 \label{lemmax}
 If {\bf{(A1)}} holds,  $\bigP_{\theta^{\star}}$ - eventually a.s. ,
 $
|Y|_{(n)}^{2}\leq 5 \sigma_{\star}^{2} \log n
 $.
 \end{lem}

\section{Algorithm and simulations}
\label{sec:algo}

In this section we first present our practical algorithm. We then apply it in the case of Gaussian emissions with unknown common variance and compare our estimator with the BIC estimator that is when we choose in (\ref{tau}) the BIC penalty $pen(n,\tau) = \frac{k-1}{2}|\tau|\log n$.

\subsection{Algorithm}

We start this section with the definition of the terms used below :

\begin{itemize}
\item
A maximal node of a complete tree $\tau$ is a string $u$ such that, for any $x$ in $\bigX$, $ux$ belongs to $\tau$. We denote by $N(\tau)$ the set of maximal nodes in the tree $\tau$.
\item
The score of a complete tree $\tau$ on the basis of the observation $(Y_{1},\ldots,Y_{n})$  is the penalized maximum likelihood associated with $\tau$ :
\begin{equation}
\label{score} sc(\tau) = -\sup_{\theta \in \Theta_{\tau}} \log g_{\theta}(Y_{1:n}) + pen( n, \tau) \end{equation}
\end{itemize}
We also require that the emission model belongs to an exponential family such that :\\[0.2cm]
(i) There exists $D \in \bigN^{\star}$, a function $s: \bigX \times \bigY \longrightarrow \bigR^{D}$ of sufficient statistic and functions $h: \bigX \times \bigY \longrightarrow \bigR$, $\psi: \Theta_{e} \longrightarrow \bigR^{D}$, and $A : \Theta_{e} \longrightarrow \bigR$, such that the emission density can be written as : \\[0.2cm]
 \[g_{\theta_{e,x},\eta}(y) = h(x,y)\exp \left[ \left< \psi(\theta_{e}), s(x,y)\right> - A(\theta_{e})\right]\]
where $\left< .,.\right>$ denotes the scalar product in $\bigR^{D}$. \\[0.2cm]
(ii) For all $S\in \bigR^{D}$, the equation : 
\[\nabla_{\theta_{e}} \psi(\theta_ {e})S - \nabla_{\theta_{e}} A(\theta_{e})=0\]
where $\nabla_{\theta_{e}} $ denotes the gradient, has a unique solution denoted by $\bar{\theta}_{e}(S)$.\\[0.4cm]
Assumption (ii) states that the function $\bar{\theta}_{e} : S \in \bigR^{D} \rightarrow \bar{\theta}_{e}(S) \in \Theta_{e}$ that returns the complete data
maximum likelihood estimator corresponding to any feasible value of the sufficient statistics is available in closed-form.\\[0.2cm]
The key idea of our algorithm is a "bottom to the top" pruning technique. Starting from the maximal complete tree of depth $M = \lfloor \log n \rfloor$, denoted by $\tau_{M}$, we change each maximal node into a leaf whenever the resulting tree decreases the score.\\[0.2cm]
We then need to compute  the maximum likelihood of any complete tree subtree of $\tau_{M}$. We start the algorithm by running several iterations of the EM algorithm. During this preliminary step we build estimators of sufficient statistics. These statistics will be used later in the computation of the maximum likelihood estimator $\hat{\theta}_{\tau} \in \Theta_{\tau}$ which realizes the supremum in (\ref{score}) for any complete context tree $\tau$ subtree of $\tau_{M}$.\\
For any $n\ge 0$, we denote by $W_{n}$ the vectorial random sequence $W_{n} = \left(X_{n-M+1},\ldots,X_{n} \right)$. For $n$ big enough, $M\ge d(\tau^{\star})$ and $(W_{n})_{n}$ is a Markov chain. The intermediate quantity (see \cite{CMR05}) needed in the EM algorithm for the HMM $(W_{n},Y_{n})$ can be written as:\\
for any $(\theta, \theta')$ in $\Theta_{\tau_{M}}$ :
\begin{eqnarray*}
Q_{\theta,\theta'}&=& E_{\theta'} (\log(g_{\theta}(W_{1:n},Y_{1:n})) \big|Y_{1:n} )\\
 &=&  E_{\theta'} ( \nu(W_{1}) \big|Y_{1:n} )  + \sum\limits_{i = 1}^{n-1}  E_{\theta'} ( \log P_{\theta_{t}}(W_{i} , W_{i+1}) \big|Y_{1:n} )  \\
 &+&  \sum\limits_{i = 1}^{n}  E_{\theta'} ( \log g_{\theta_{e,W_{i,M},\eta}}(Y_{i}) \big|Y_{1:n} ).  
\end{eqnarray*}\\[0.2cm]
Notice, for any $\theta \in \Theta_{\tau_{M}}$, if $(w,w') \in (\bigX^{M})^{2}$ are such that $ w_{2:M} \neq w'_{1:M-1}$, then  $P_{\theta_{t}}(w,w') =  0$. \\[0.2cm]
For any $w \in \bigX^{M}$ and any $w' \in \bigX^{M}$ if we denote by 
\begin{eqnarray*}
& &\forall \ i=1,\ldots,n , \ \Phi_{i|n}^{\theta'} (w ) = P_{\theta'}(W_{i} = w|Y_{1:n}),\\
& &\forall \ i=1,\ldots, n-1 , \ \Phi_{i:i+1|n}^{\theta'} (w,w' ) = P_{\theta'}(W_{i} = w,W_{i+1} = w' |Y_{1:n}),
 \end{eqnarray*}
and
\begin{eqnarray*}
 S^{\theta'}_{t,n} &=&\left(\dfrac{ \left(\sum\limits_{i = 1}^{n-1} \Phi_{i:i+1|n}^{\theta'}(w,w')\right)}{n} \right)_{(w,w') \in \bigX^M}\\[0.4cm]
S^{\theta'}_{e,n} &= & \dfrac{1}{n}\sum\limits_{x \in \bigX}\sum\limits_{i = 1}^{n} \left(\sum\limits_{w\in \bigX^M  |  w_{M} = x} \Phi_{i|n}^{\theta'} (w) \right)  s(x,Y_{i})
\end{eqnarray*}
then there exists a function $C$ such that :
\begin{equation}
\label{qte intermed}
\dfrac{1}{n}Q_{\theta,\theta'} = \dfrac{1}{n} C(\theta',Y_{1:n}) + \left< S^{\theta'}_{t,n}, \log P_{\theta_{t}}\right> +\left< S^{\theta'}_{e,n},\psi(\theta_{e})\right>-A(\theta_{e}) .
\end{equation}

If, for some complete tree $\tau$, we restrict $\theta_{t}$ in $\Theta_{t,\tau}$, then for any $s$ in $\tau$, for any $ w $ in $ \bigX^M$ such that $s$ is postfix of $w$, for any $x$ in $\bigX$, we have $P_{\theta_{t}}\left(w,(w_{2:M} x)\right) = P_{s,x}(\theta_{t})$.\\[0.4cm]
Thus, the vector $P_{s,.}$ maximising this equation is solution of the Lagrangian, \\[0.2cm]
$ \left\{\begin{array}{lll}
\dfrac{\delta}{\delta P_{s,x} } & \left[ \dfrac{1}{n}Q_{\theta,\theta'} + \Lambda(\sum\limits_{x' \in \bigX}P_{s,x'} -1 ) \right] & =0 \ , \ \forall x \in \bigX\\[0.2cm]
\dfrac{\delta}{\delta \Lambda } & \left[ \dfrac{1}{n}Q_{\theta,\theta'} + \Lambda(\sum\limits_{x' \in \bigX}P_{s,x'} -1 ) \right] &  =0
\end{array}\right.$\\[0.2cm]
and, finally, the estimator of $\theta_{t} \in \Theta_{t,\tau}$ maximising the quantity $Q(\theta',.)$ only depends on the sufficient statistic $S^{\theta'}_{t,n}$ and is given by : 
\begin{equation} \label{P chapeau s}
 \bar{P}_{s,x}(S^{\theta'}_{t,n}) =  \dfrac{  \sum\limits_{w \in \bigX^M |\ s \  \mathrm{postfix} \ of \ w } S^{\theta'}_{t,n} (w,(w_{2:M}x))}{ \sum\limits_{x' \in \bigX } \quad \sum\limits_{w\in \bigX^M |\ s \  \mathrm{postfix} \ of \ w }S^{\theta'}_{t,n}( w,(w_{2:M}x'))} .
\end{equation}
\begin{algorithm}[h]
\caption{Preliminary computation of the sufficient statistics}
\label{algosuffstat}
\begin{algorithmic}[1]
\REQUIRE $\theta_{0} = (\theta_{t,0},\theta_{e,0}) \in \Theta_{\tau_{M}}$ be an initial value for the parameter $\theta$.
\REQUIRE Let $t_{EM}$ be a threshold.
\STATE $stop=0$
\STATE $i=0$
\WHILE{($stop = 0$)}
\STATE $i=i+1$
\STATE \underline{M step} : compute the quantities $S^{\theta_{i-1}}_{t,n}$ and $S^{\theta_{i-1}}_{e,n}$ 
\STATE \underline{E step} : set $$\theta_{i} = \left( \left(\bar{P}_{w,x}(S^{\theta_{i-1}}_{t,n})\right)_{w,x} \ , \  \bar{\theta}_{e}(S^{\theta_{i-1}}_{e,n})\right)$$
\IF{($||\theta_{i} - \theta_{i-1}|| < t_{Em}$)}
\STATE $stop = 1$
\ENDIF
\ENDWHILE
\STATE \underline{M step} : compute the quantities $S^{\theta_{i}}_{t,n}$ and $S^{\theta_{i}}_{e,n}$ 
\STATE $S_{t} = S^{\theta_{i}}_{t,n}$ and $S_{e} = S^{\theta_{i}}_{e,n}$ 
\RETURN $(S_{t}, S_{e})$
\end{algorithmic}
\end{algorithm}

While Algorithm \ref{algosuffstat} computes the sufficient statistics $S_{t}$ and $S_{e}$ on the basis of the observations $(Y_{k})_{k \in \left\{ 1,\ldots,n\right\} }$, Algorithm \ref{algoprun} is our pruning Algorithm. This algorithm begins with the estimation of the exhaustive statistics calling Algorithm \ref{algosuffstat}. As Algorithm \ref{algosuffstat} is prone to the convergence towards local maxima, we set our initial parameter value $\theta_{0}$ after running a preliminary \textit{k-means}  algorithm (see \cite{McQueen}, \cite{inaba}): we assign the values $Y_{1:n}$ into $k$ clusters which produces a sequence of "clusters" $\tilde{X}_{1:n}$. A first estimation of the emission parameters is then possible using this clustering, the initial transition parameter $\theta_{0,t} = \left(P^{0}_{w,i}\right)_{w \in \bigX^{M}, \ i \in \bigX}$ is also computed on the basis of the sequence $\tilde{X}_{1:n}$ using the relation : 
 \begin{equation*}\forall w \in \bigX^{M}, \ \forall x \in \bigX,  P^{0}_{w,x} =  \frac{\sum\limits_{i=1}^{n-M} \one_{ \tilde{X}_{i:i+M-1} =w}\one_{\tilde{X}_{i+M}=x} }{\sum\limits_{i=1}^{n-M} \one_{ \tilde{X}_{i:i+M-1} =w}}   \ . \end{equation*}
Then, starting with the initialisation $\tau = \tau_{M}$, we consider, one after the other, the maximal nodes $u$ of $\tau$. We build a new tree $\tau_{\text{test}}$ by taking  out of $\tau$ all the contexts $s$ having $u$ as postfix and adding $u$ as a new context: $\tau_{test} = \tau \setminus \left\{ ux \big|ux \in \tau, \ x \in \bigX \right\} \bigcup \left\{u \right\}  $. Let $\hat{\theta}_{\text{test}} = \left(( (\bar{P}_{s,x}(S_{t}))_{s \in \tau_{\text{test}}, x \in \bigX} , \bar{\theta}_{e}(S_{e})\right)$ which, hopefully, becomes an acceptable proxy for $ \argmax\limits_{\theta \in \Theta_{\tau_{test}}} \log g_{\theta}(Y_{1:n})$. Let $- \log g_{\hat{\theta}_{\text{test}}}(Y_{1:n}) + pen( n, \tau_{\text{test}})$ be an approximation of the score of the context tree $\tau_{\text{test}}$ still denoted by $sc(\tau_{\text{test}}) $, then, if $sc(\tau_{\text{test}}) <sc(\tau)$, we set $\tau = \tau_{\text{test}}$. In Algorithm \ref{algoprun}, the role of $\tau_{2}$ is to insure that all the branches of $\tau$ are tested before shortening again a branch already tested.

\begin{algorithm}[h]
\caption{Bottom to the top pruning algorithm}
\label{algoprun}
\begin{algorithmic}[1]
\REQUIRE Let $t_{EM}$ a threshold.
\STATE Compute $(S_{t}, S_{e})$ with  Algorithm \ref{algosuffstat} with the $t_{EM}$ threshold.
\STATE $\hat{\theta} = \left(\left(\bar{P}_{w,x}(S_{t})\right)_{w \in \tau_{M}, x \in \bigX},\bar{\theta}_{e}(S_{e}) \right) $
\STATE\textit{Pruning procedure} :
\STATE $\tau=\tau_{2}=\tau_{M}$
\STATE $change=YES$
\WHILE {($change = YES$ AND $|\tau|\ge1$)}
\STATE $change=NO$
	\FOR {($u \in N(\tau)$)}
	\IF{($u \in N(\tau_{2})$)}
		\STATE $L_{u}(\tau_{2}) =\left\{s \in \tau_{2} | u \ postfix \ of \ s\right\}$
		\STATE $\tau_{\text{test}} =\left[ \tau_{2} \setminus L_{u}(\tau_{2}) \right] \bigcup \left\{ u \right\}$
		\STATE $ \hat{\theta}_{\text{test}} = \left(\left(\bar{P}_{s,x}(S_{t})\right)_{s\in \tau, x \in \bigX},\bar{\theta}_{e}(S_{e}) \right) $ 
		
		\IF {($sc(\tau_{\text{test}}) < sc(\tau_{2}) $)}
			
			\STATE $\tau_{2}=\tau_{\text{test}}$
			\STATE $\hat{\theta} = \hat{\theta}_{\text{test}}$
			\STATE $change =YES$
		\ENDIF
	\ENDIF	
	\ENDFOR
\STATE $\tau = \tau_{2}$	
\ENDWHILE
\RETURN $\tau$
\end{algorithmic}
\end{algorithm}

\subsection{Simulations}

We propose to illustrate the a.s convergence of $\hat{\tau}_{n}$ using Algorithm \ref{algoprun} in the case of Gaussian emission with unknown variance. We set $k=2$, and use as minimal complete context tree one of the two complete trees represented in Figure \ref{taustar1} and Figure \ref{taustar2}. The true transitions probabilities associated with each trees are indicated in boxes under each context.
\begin{figure}[!h]
\begin{center}
\includegraphics[width=0.45\textwidth]{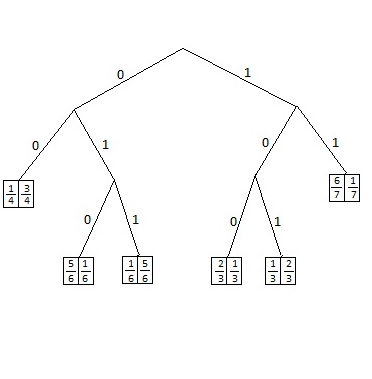}
\end{center}
\caption{Graphic representation of the complete context tree $\tau^{\star}_{1}$ with transition probabilities indicated in the box under each leaf $s$: $P_{s,0}^{\star} \ | \ P_{s,1}^{\star}$ } 
\label{taustar1} 
\end{figure}
\begin{figure}[!h]
\begin{center}
\includegraphics[width=0.35\textwidth]{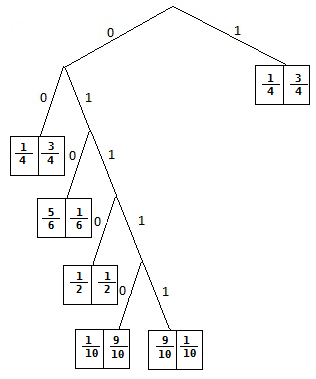}
\end{center}
\caption{Graphic representation of the complete context tree $\tau^{\star}_{2}$ with transition probabilities indicated in the box under each leaf $s$: $P_{s,0}^{\star} \ | \ P_{s,1}^{\star}$ } 
\label{taustar2} 
\end{figure}

For each tree $\tau^{\star}_{1}$ and $\tau^{\star}_{2}$, we will simulate 3 samples of the VLHMM, choosing as true emission parameters $m_{0}^{\star} =0$, $\sigma^{2,\star} = 1$ and $m_{1}^{\star}$ varying in $\{2,3,4\}$. In the preliminary EM steps, we use as threshold $t_{EM}=0.001$\\[0.2cm]
\begin{table}[!h]
\begin{center}
\begin{tabular}{|c|c|c|c||c|c|c|}
\hline
 \multicolumn{7}{|c|}{ $\tau^{\star} =\tau^{\star}_{1}$, $|\tau^{\star}|=6$   }        \\
\hline
 $ $ & \multicolumn{3}{|c||}{Penalty (\ref{pen})}  & \multicolumn{3}{c|}{BIC penalty }\\
 
\cline{2-7}
\textbf{$n/m_{1}^{\star}$} & $2$ & $3$ & $4$ & $2$ & $3$ & $4$\\
\hline
\hline
100 & 2 & 2 & 2   & 2 & 3 & 3 \\
\hline
1000 & 2 & 2 & 2  & 7 & \textbf{6} &\textbf{6} \\
\hline
2000 & 2 & 2 & 4  & \textbf{6} & \textbf{6} & \textbf{6}\\
\hline
5000 & 2 & 4 & 4  &  7 &\textbf{6} & \textbf{6}\\
\hline
10000 & 4 & \textbf{6} & \textbf{6}  &  7 & \textbf{6} &\textbf{6}\\
\hline
20000 & 5 & \textbf{6}& \textbf{6}  &\textbf{6} & \textbf{6} & \textbf{6}\\
\hline
30000 & 5 & \textbf{6} & \textbf{6}  & \textbf{6} & \textbf{6}& \textbf{6}\\
\hline
40000 & \textbf{6}& \textbf{6} & \textbf{6}  & 7 & \textbf{6} & \textbf{6}\\
\hline
50000 & \textbf{6} & \textbf{6} & \textbf{6}  & 7 & \textbf{6} & \textbf{6}\\
\hline
\end{tabular} 
\end{center}
\caption{Case $\tau^{\star} =\tau_{1}^{\star}$. Comparison of $|\hat{\tau}_{n}|$ between our estimator and the BIC estimator for different values of $n$ and $m_{1}^{\star}$.}
\label{cardtau1}
\end{table}

\begin{table}[!h]
\begin{center}
\begin{tabular}{|c|c|c|c||c|c|c|}
\hline
 \multicolumn{7}{|c|}{ $\tau^{\star} =\tau^{\star}_{1}$, $|\tau^{\star}|=6$   }        \\
\hline
 $ $ & \multicolumn{3}{|c||}{Penalty (\ref{pen})}  & \multicolumn{3}{c|}{BIC penalty }\\
 
\cline{2-7}
\textbf{$n/m_{1}^{\star}$} & $2$ & $3$ & $4$ & $2$ & $3$ & $4$\\
\hline
\hline
100 & -202 & -202 &  -190 & -6  & -6  &  2 \\
\hline
1000 & -235 & -213&-155 & 4   & -2    & 25 \\
\hline
2000 & -221& -129 &  -88 & 8 & -4  & 4\\
\hline
5000 & -144 & -36 &  -20 &  5   & -4   &  -5\\
\hline
10000 & -75 & -5 &  -4  & 4   & -5    & -4\\
\hline
20000 & -6& -4 & -4 &  10    & -4    &  -4\\
\hline
30000 & 21 & -5 &  -4 & 10    & -5    &  -4\\
\hline
40000 & 12 & -4 &   -3 & 10     & -4   &   -3\\
\hline
50000 & 12 & -7 &   -4  & 10     & -4    & -4\\
\hline
\end{tabular} 
\end{center}
\caption{Case $\tau^{\star} =\tau_{1}^{\star}$. Score difference $sc(\hat{\tau}_{n})-sc(\tau^{\star} )$.}
\label{scoretau1}
\end{table}

\begin{table}[!h]
\begin{center}
\begin{tabular}{|c|c|c|c||c|c|c|}
\hline
 \multicolumn{7}{|c|}{ $\tau^{\star} =\tau^{\star}_{2}$, $|\tau^{\star}|=6$   }        \\
\hline
 $ $ & \multicolumn{3}{|c||}{Penalty (\ref{pen})}  & \multicolumn{3}{c|}{BIC penalty }\\
 
\cline{2-7}
\textbf{$n/m_{1}^{\star}$} & $2$ & $3$ & $4$ & $2$ & $3$ & $4$\\
\hline
\hline
100 & 2 & 2 & 2   & 2 & 2& 2 \\
\hline
1000 & 2 & 2 & 2  & 3 & \textbf{6} &\textbf{6} \\
\hline
2000 & 2 & 2 & 2  & \textbf{6} & \textbf{6} & \textbf{6}\\
\hline
5000 & 2 & 3 & 3  & \textbf{6} & \textbf{6} & \textbf{6}\\
\hline
10000 & 3 & 3 & 3  & \textbf{6} & \textbf{6} & \textbf{6}\\
\hline
20000 & 3 & 3 & \textbf{6} &\textbf{6} & \textbf{6} & \textbf{6}\\
\hline
30000 & 3 & 3 & \textbf{6}   & \textbf{6} & \textbf{6}& \textbf{6}\\
\hline
40000 &3 & \textbf{6} & \textbf{6}   & \textbf{6} & \textbf{6} & \textbf{6}\\
\hline
50000 & 3 & \textbf{6} & \textbf{6}    & \textbf{6} & \textbf{6} & \textbf{6}\\
\hline
\end{tabular} 
\end{center}
\caption{Case $\tau^{\star} =\tau_{2}^{\star}$. Comparison of $|\hat{\tau}_{n}|$ between our estimator and the BIC estimator for different values of $n$ and $m_{1}^{\star}$.}
\label{cardtau2}
\end{table}

\begin{table}[!h]
\begin{center}
\begin{tabular}{|c|c|c|c||c|c|c|}
\hline
 \multicolumn{7}{|c|}{ $\tau^{\star} =\tau^{\star}_{2}$, $|\tau^{\star}|=6$   }        \\
\hline
 $ $ & \multicolumn{3}{|c||}{Penalty (\ref{pen})}  & \multicolumn{3}{c|}{BIC penalty }\\
\cline{2-7}
\textbf{$n/m_{1}$} & $2$ & $3$ & $4$ & $2$ & $3$ & $4$\\
\hline
\hline
100 & -201 &  -202 &   -195  &  -10  &  -6   &   1 \\
\hline
1000 & -266 &  -246 &-229  & 5   &   -1   &  -2 \\
\hline
2000 & -272 & -239 &    67 & 4   &  -1   &  324\\
\hline
5000 & -272 & -200&  -151  &  2  &  -2    &  -5\\
\hline
10000 & -242 & -128 &  -52  &  6   & -2    &   -4\\
\hline
20000 &   -227 & 12&  -6&   6   &  -6  &   -6\\
\hline
30000 &  -191 & 141 &  -6  &7    &  -5   & -6\\
\hline
40000 &  -159 &  -6 &   -8 &  8    &-6   &    -8\\
\hline
50000 &   -136 &  -6 &  -9  &  7     & -6  &   -8\\
\hline
\end{tabular} 
\end{center}
\caption{Case $\tau^{\star} =\tau_{2}^{\star}$. Score difference $sc(\hat{\tau}_{n})-sc(\tau^{\star} )$.}
\label{scoretau2}
\end{table}

The results of our simulations are summarized in Tables \ref{cardtau1} to \ref{scoretau2}. The size of the estimated tree $|\hat{\tau}_{n}|$ for different values of $n$ and $m_{1}^{\star}$ are noticed in  Table \ref{cardtau1} when $\tau^{\star} = \tau^{\star}_{1}$ (resp. in the table Figure \ref{cardtau2} when $\tau^{\star} = \tau^{\star}_{2}$) for the two choices of penalties $ pen_{\alpha}(n,\tau) = \sum\limits_{t=1}^{|\tau|} \dfrac{(k-1)t + \alpha}{2} \log n$ with $\alpha = 5.1$ and $ pen(n,\tau) =\dfrac{k-1}{2}|\tau| \log n $. The first important remark we make regarding Tables \ref{cardtau1} and  \ref{cardtau2} is that, on each simulation and whatever the penalty we used, when $|\hat{\tau}_{n}|=|\tau^{\star}|$ we also had $\hat{\tau}_{n}=\tau^{\star}$, in the same way, each time $|\hat{\tau}_{n}|<|\tau^{\star}|$ (resp. $|\hat{\tau}_{n}|>|\tau^{\star}|$  ), $\hat{\tau}_{n}$ was a subtree of  $\tau^{\star}$ (resp. $\tau^{\star}$ was a subtree of $\hat{\tau}_{n}$). 
For any combination of $\tau^{\star}$ and $m_{1}^{\star}$, both estimators seem to converge, except our estimator in the case $\tau^{\star}=\tau_{2}^{\star}$ and $m_{1}^{\star}=2$, where 50 000 measures is not enough to reach the convergence. However, for small samples, smaller models are systematically chosen with our estimator, while the BIC estimator is reaching the right model for relatively small samples. This behaviour of our estimator shows that our penalty is too heavy.

The score differences $ sc(\hat{\tau}_{n})- sc(\tau^{\star})$ Table \ref{scoretau1} when $\tau^{\star} = \tau^{\star}_{1}$ and Table \ref{scoretau2} when $\tau^{\star} = \tau^{\star}_{2}$ are the differences between the score of $\hat{\tau}_{n}$ computed with the estimated parameter $\hat{\theta}_{n}$ and the score of $\tau^{\star}$ computed with the the real parameters. These informations allow us to know when the estimators $\hat{\tau}_{n},\hat{\theta}_{n}$ are well estimated by Algorithm \ref{algoprun}. Indeed, when $\hat{\tau}_{n} \neq \tau^{\star}$, if the score of  $\tau^{\star}$ computed with the real transition and emission parameters is smaller than the score of our estimator with estimated parameters (non negative score difference), then the estimator given by Algorithm \ref{algoprun} is not the expected estimator defined by (\ref{tau}). In particular, Table \ref{scoretau1} shows that the over estimation of the BIC estimator in the case $m_{1}^{\star} = 2$ (Table \ref{scoretau1}) can be due to a local minima problem: Algorithm \ref{algoprun} selected a tree $\tau$ such that $|\tau|>|\tau^{\star }|$ whereas $\tau^{\star }$ had a smaller score. This problem might occur because we use an EM type algorithm which often leads to local minima. Although we try to take an initial value of the parameters in a neighbourhood of the real ones using the preliminary k-means algorithm, this problem persists. Extra EM loops for each tested tree in Algorithm \ref{algoprun} could also provide a better estimation of the parameters and then improve the score estimation for each tested tree, but it would also increase the complexity of the algorithm.

Finally, we observe that bigger the quantity $\vert m_{0}^{\star}-m_{1}^{\star} \vert$ is, quicker the convergence of our estimator or BIC estimator occurs. This phenomenon can be easily understood as very different emission distributions for different states leads to an easier estimation of the underlying state sequence on the basis of the observations and allows us to build a more precise description of the VLMC behaviour.

\section{Conclusion}

In this paper, we were interested in the statistical analysis of Variable Length Hidden Markov Models (VLHMM). 
We have presented such models then we estimated the context tree of the hidden process using penalized maximum likelihood. We have shown how to choose the penalty so that the estimator is strongly consistent without any prior upper bound on the depth or on the size of the context tree of the hidden process. We have proved that our general consistency theorem applies when the emission distributions are Gaussian with unknown means and the same unknown variance. We have proposed a pruning algorithm and have applied it to simulated data sets. This illustrates the consistency of our estimator, but also suggests that smaller penalty could lead to consistent estimation.\\
Finding the minimal penalty insuring the strong consistency of the estimator with no prior upper bound remains unsolved. A similar problem has been solved by R. van Handel \cite{Han09} to estimate the order of finite state Markov chains, and by E. Gassiat and R. van Handel \cite{HanGas} to estimate the number of populations in a mixture with i.i.d. observations. The basic idea is that the maximum likelihood behaves as the maximum of approximate chi-square variables, and that the behavior of the maximum likelihood statistic may be investigated using empirical process theory tools to obtain a $\log \log n$ rate of growth. However, it is known for HMM that the maximum likelihood does not behave this way and converges weakly to infinity, see  \cite{EGCK00}. 
We did by-pass the problem by using information theoretic inequalities, but
understanding the pathwise fluctuations of the likelihood in HMM models remains a difficult problem to be solved.

\appendices
\label{sec:proofs}
\section{Proof of Lemma \ref{lemsinf} }

For any  partition $(I_{i},\ldots,I_{k})$ of $\bigR$ in $k$ intervals,
\begin{eqnarray*}
 \sigma_{I_{i},\ldots,I_{k}}^{2}&=&\sum\limits_{j=1}^{k} \bigP_{\theta^{\star}}(Y_{1} \in I_{k}) Var_{\theta^{\star}}(Y_{1}|Y_{1} \in I_{k}) \\
  &\geq& \frac{1}{k} \inf_{I:\bigP_{\theta^{\star}}(Y \in I)\geq \frac{1}{k}}  Var_{\theta^{\star}}(Y_{1}|Y_{1} \in I)
\end{eqnarray*}

where the infimum is over all intervals $I$ of $\bigR$.
The distribution of $Y_{1}$ is the Gaussian mixture with density $g^{\star}=\sum\limits_{x \in \bigX} \pi^{\star}(x) \phi_{m_{x}^{\star},\sigma_{\star}^{2}}$,
 where $\pi^{\star}$ is the stationary distribution of $(X_{n})_{n\geq 0}$ and $\phi_{m_{x}^{\star},\sigma_{\star}^{2}}$ is the density of the normal distribution with mean $m_{x}^{\star}$ and variance $\sigma_{\star}^{2}$. The repartition function $F^{\star}$ of the distribution of $Y_{1}$ is continuous and increasing, with continuous and increasing inverse quantile function. Thus,
 $$
 \inf_{_{I_{i},\ldots,I_{k}}} \sigma_{I_{i},\ldots,I_{k}}^{2} \geq \inf_{ \overset{-\infty \leq a < b \leq +\infty :}{ F^{\star}(a)+\frac{1}{k}\leq F^{\star}(b)}}
 Var_{\theta^{\star}}(Y_{1}|Y_{1} \in ]a,b[).
 $$
But $ Var_{\theta^{\star}}(Y_{1}|Y_{1} \in ]a,b[)$ is a continuous function of $(a,b)$, and the infimum at the righ-hand side of the inequality is attained at some $(\overline{a},\overline{b})$ (eventually infinite) such that $F^{\star}(\overline{a})+\frac{1}{k}\leq F^{\star}(\overline{b})$. Thus
$ Var_{\theta^{\star}}(Y_{1}|Y_{1} \in ]\overline{a},\overline{b}[)>0$, and $s_{inf}>0$.\\
For any  partition $(I_{i},\ldots,I_{k})$ of $\bigR$ in $k$ intervals,
\begin{equation*}
\hat{\sigma}_{I_{i},\ldots,I_{k}}^{2}(Y_{1:n}) -  \sigma_{I_{i},\ldots,I_{k}}^{2}= \frac{1}{n} \sum\limits_{i=1}^{n}Y_{i}^{2} -E(Y_{1}^{2}) - \sum\limits_{j=1}^{k} \left(  \frac{(\sum\limits_{i = 1}^{n} Y_{i} \one_{I_{j}} (Y_{i}))^{2}}{n^{2}} \frac{n}{\sum\limits_{i = 1}^{n}  \one_{I_{j}}(Y_{i})} - \frac{E(Y\one_{I_{j}} (Y_{1}))^{2}}{E(\one_{I_{j}}(Y_{1}))}\right)
\end{equation*}
 so that 
\begin{equation*}
 \sup\limits_{I_{1},\ldots,I_{k}} \left\vert \hat{\sigma}_{I_{i},\ldots,I_{k}}^{2}(Y_{1:n}) -  \sigma_{I_{i},\ldots,I_{k}}^{2}\right\vert  \leq \frac{1}{n} \left\vert \sum\limits_{i=1}^{n}Y_{i}^{2} -E(Y_{1}^{2}) \right\vert
 + k\sup\limits_{I \ \text{interval of}\ \bigR}\Bigg\vert \dfrac{(\sum\limits_{i = 1}^{n} Y_{i} \one_{I} (Y_{i}))^{2}}{n^{2}} \frac{n}{\sum\limits_{i = 1}^{n}  \one_{I}(Y_{i})} - 
 \frac{E(Y_{1}\one_{I} (Y_{1}))^{2}}{E(\one_{I}(Y_{1}))}\Bigg\vert.
\end{equation*}
Using \cite{L92}, $(Y_{n})_{n\geq 0}$ is a stationary ergodic process, so that
$\frac{1}{n} \sum\limits_{i=1}^{n}Y_{i}^{2} -E(Y_{1}^{2})$ tends to 0 $\bigP_{\theta^{\star}}$ a.s. 
Let $\epsilon >0$. We now consider separately  the intervals $I$ such that $E(\one_{I}(Y)) \le \epsilon$ or $E(\one_{I}(Y)) > \epsilon$.\\
$ \bullet$ Let $I $ be such that $E(\one_{I}(Y_{1})) \le \epsilon$.\\
Using Cauchy Schwarz inequality,
\begin{eqnarray*}
\left(\frac{1}{n} \sum Y_{i} \one_{I}(Y_{i})\right)^{2} &\le& \left(\frac{1}{n} \sum Y_{i}^{2} \one_{I}(Y_{i}) \right) \times \left(\frac{1}{n} \sum  \one_{I}(Y_{i})\right),\;\\
E\left(Y_{1}\one_{I}(Y_{1})\right)^{2} &\le& E\left(Y_{1}^{2}\one_{I}(Y_{1})\right) E\left(\one_{I}(Y_{1})\right)\;
\end{eqnarray*}
and,
\begin{equation*}
E\left(Y_{1}^{2} \one_{I}(Y_{1})\right) \le \sqrt{E(Y_{1}^{4})} \sqrt{E(\one_{I}(Y_{1}))} \leq M \sqrt{\epsilon}
\end{equation*}
for some fixed positive constant $M$.
Thus,
\begin{eqnarray*}
& &\left| \frac{(\sum_{i=1}^{n} Y_{i} \one_{I} (Y_{i}))^{2}}{n^{2}} \frac{n}{\sum_{i=1}^{n} \one_{I}(Y_{i})} - \frac{E(Y_{1}\one_{I} (Y_{1}))^{2}}{E(\one_{I}(Y_{1}))}\right|\\
&\leq &\frac{1}{n} \sum_{i=1}^{n} Y_{i}^{2} \one_{I}(Y_{i}) +  E(Y_{1}^{2}\one(Y_{1})) \\
& \leq & \left|\frac{1}{n} \sum_{i=1}^{n} Y_{i}^{2} \one_{I}(Y_{i})- E(Y_{1}^{2}\one(Y_{1}))\right| + 2 E(Y_{1}^{2}\one(Y_{1}))\\
& \leq & \left|\frac{1}{n} \sum_{i=1}^{n} Y_{i}^{2} \one_{I}(Y_{i})- E(Y_{1}^{2}\one(Y_{1}))\right| +2M\sqrt{\epsilon}.
\end{eqnarray*}
$ \bullet$ Let now $I $ be such that $E(\one_{I}(Y_{1})) > \epsilon$.\\
\begin{eqnarray*}
& & \left| \frac{(\sum_{i=1}^{n} Y_{i} \one_{I} (Y_{i}))^{2}}{n^{2}} \frac{n}{\sum_{i=1}^{n}  \one_{I}(Y_{i})} - \frac{E(Y_{1}\one_{I} (Y_{1}))^{2}}{E(\one_{I}(Y_{1}))}\right|\\
&=&\left|\frac{ \sum_{i=1}^{n}  Y_{i} \one_{I} (Y_{i})}{n} \frac{1}{\sqrt{\dfrac{\sum_{i=1}^{n}  \one_{I} (Y_{i})}{n}}} -  \frac{E(Y_{1}\one_{I} (Y_{1}))}{\sqrt{E(\one_{I}(Y_{1}))}}\right|\\
&\times & \bigg|\frac{ \sum_{i=1}^{n}  Y_{i} \one_{I} (Y_{i})}{n} \frac{1}{\sqrt{\frac{\sum_{i=1}^{n}  \one_{I} (Y_{i})}{n}}} +  \frac{E(Y_{1}\one_{I} (Y_{1}))}{\sqrt{E(\one_{I}(Y_{1}))}}\bigg|\\
& \leq &  \Bigg[ \left|\frac{ \sum_{i=1}^{n}  Y_{i} \one_{I} (Y_{i})}{n}\right|\left|  \frac{1}{\sqrt{\frac{\sum_{i=1}^{n}  \one_{I} (Y_{i})}{n}}} - \frac{1}{\sqrt{E(\one_{I}(Y_{1}))}} \right|\\
&+ & \left|\frac{\frac{ \sum_{i=1}^{n}  Y_{i} \one_{I} (Y_{i})}{n}-E(Y_{1}\one_{I} (Y_{1}))}{\sqrt{E(\one_{I}(Y_{1}))}}\right| \Bigg] \\
&\times & \Bigg[ \left|\frac{ \sum_{i=1}^{n}  Y_{i} \one_{I} (Y_{i})}{n}\right| \left|  \frac{1}{\sqrt{\frac{\sum_{i=1}^{n}  \one_{I} (Y_{i})}{n}}}+ \frac{1}{\sqrt{E(\one_{I}(Y_{1}))}} \right| \\
&+& \left|\frac{\frac{ \sum_{i=1}^{n}  Y_{i} \one_{I} (Y_{i})}{n}-E(Y\one_{I} (Y_{1}))}{\sqrt{E(\one_{I}(Y_{1}))}}\right| \Bigg] \\
&\leq &  \Bigg[\bigg( \frac{\sum_{i=1}^{n}|Y_{i}|}{n} \bigg)\frac{ \left| \sqrt{\frac{\sum_{i=1}^{n}  \one_{I} (Y_{i})}{n}} -\sqrt{E(\one_{I}(Y_{1}))}\right|}{\epsilon} \\
&+&    \frac{\left|\frac{ \sum_{i=1}^{n}  Y_{i} \one_{I} (Y_{i})}{n} -E(Y_{1}\one_{I} (Y_{1}))\right|}{\sqrt{\epsilon}}\Bigg]\\
& \times & \left[ \frac{2}{\epsilon}\frac{\sum_{i=1}^{n}  |Y_{i}|}{n} + \frac{\left|\frac{ \sum_{i=1}^{n}  Y_{i} \one_{I} (Y_{i})}{n} -E(Y_{1}\one_{I} (Y_{1}))\right|}{\sqrt{\epsilon}}\right]
\end{eqnarray*}

Now, using Lemma \ref{lemunif} below, one gets that, for all positive $\epsilon$,
\begin{equation*}
\limsup_{n \to \infty} \sup\limits_{I \ \text{interval of}\  \bigR} \Bigg| \frac{E(Y_{1}\one_{I} (Y_{1}))^{2}}{E(\one_{I}(Y_{1}))}- \frac{(\sum_{i=1}^{n} Y_{i} \one_{I} (Y_{i}))^{2}}{n^{2}} \frac{n}{\sum_{i=1}^{n} \one_{I}(Y_{i})} \Bigg| 
\le 2M\sqrt{\epsilon} 
\end{equation*}
$\bigP_{\theta^{\star}}$-a.s.  so that
\begin{equation*}
\lim_{n \to \infty} \sup\limits_{I \ \text{interval of}\  \bigR} \Bigg| \frac{E(Y_{1}\one_{I} (Y_{1}))^{2}}{E(\one_{I}(Y_{1}))}- \frac{(\sum_{i=1}^{n} Y_{i} \one_{I} (Y_{i}))^{2}}{n^{2}} \frac{n}{\sum_{i=1}^{n} \one_{I}(Y_{i})} \Bigg| 
=0
\end{equation*}
$\bigP_{\theta^{\star}}$-a.s.  and the Lemma follows.
\begin{lem}
\label{lemunif}
$\sup_{I} \left|\frac{1}{n} \sum Y_{i}^{2} \one_{I}(Y_{i}) - E\left(Y_{1}^{2} \one_{I}(Y)\right) \right|$,   
$\sup_{I} \left|\frac{1}{n} \sum Y_{i} \one_{I}(Y_{i}) - E\left(Y_{1} \one_{I}(Y_{1})\right)\right|$ and
$\sup_{I} \left|\frac{1}{n} \sum\one_{I}(Y_{i}) - E\left( \one_{I}(Y_{1}) \right)\right|$ (where the supremum is over all intervals $I$ in $\bigR$) tend to $0$ as $n$ tends to infinity,  $\bigP_{\theta^{\star}}$ a.s.
\end{lem}

\begin{proof}
Let us note  $\mathcal{F}_{a} = \left\{ x \rightarrow x^{a} \one_{I}(x) \; : \; I \  \text{interval of} \; \bigR  \right\}$
for $a=0,1,2$. Since the sequence of random variables $(Y_{n})_{n\geq 0}$ is stationary and ergodic, it is enough to prove that, for $a=0,1,2$, for any positive $\epsilon$,
there exists a finite set of functions  $\tilde{\mathcal{F}}_{a} $ such that for any $f\in\mathcal{F}_{a}$, there exists $ l , u$ in  $\tilde{\mathcal{F}}_{a}$ such that
 $l \le f \le u$ and  $E(u(Y_{1}) - l(Y_{1})) \le \epsilon$.\\
For the cases a=0 or 2 and for any positive $\epsilon$, there exist real numbers :  $L^{1}_{a,\epsilon}$ and $L^{2}_{a,\epsilon}$ such that 
$\int_{-\infty}^{L^{1}_{a,\epsilon}} x^{a} g^{\star}(x) dx \le \epsilon$ and
$\int_{L^{2}_{a,\epsilon}}^{+\infty} x^{a} g^{\star}(x) dx \le \epsilon$,
and there exists real numbers $x_{a,1}=L^{1}_{\epsilon}<x_{a,2}<....<x_{a,N_{a,\epsilon}-2}<L^{2}_{a,\epsilon}=x_{a,N_{a,\epsilon}-1}$  such that 
$\int_{x_{a,i}}^{x_{a,i+1}} x^{a} g^{\star}(x) dx < \epsilon/2$, $i=1,\ldots,N_{a,\epsilon}-2$.
Then we define 
\begin{itemize}
\item
$ I^{1}_{N_{\epsilon}} = \bigR$,
\item
for any $i = 1,...,N_{a,\epsilon}$,  $I^{1}_{a,i} =  \left[ -\infty \ , \ x_{a,i}\right]$
\item
and for any $i = 1,...,N_{a,\epsilon}$,  $I^{2}_{a,i} =  \left[ x_{a,i}\ , \ \infty\right]$
\end{itemize}
so that if $\mathcal{I}_{a}$ is the set $\mathcal{I}_{a} = \left\{I_{a,i}^{j} |i =1,...,N_{a,\epsilon}, \ j=1,2\right\}  \bigcup \left\{ \left[x_{a,i_{1}}, x_{a,i_{2}} \right]\right\}_{i_{1}<i_{2}}$ 
the set $\tilde{\mathcal{F}}_{a}= \left\{ x^{a} \one_{I} | I \in \mathcal{I}_{a} \right\}  $ verifies the above conditions.\\
For the case $a=1$ the construction of the sequence $x_{a,1}=L^{1}_{\epsilon}<x_{a,2}<....<x_{N_{a,\epsilon }-2}<L^{2}_{a,\epsilon}=x_{N_{a,\epsilon}-1}$ is such that $\int_{x_{i}}^{x_{i+1}} |x|g^{\star}(x) dx < \epsilon/2$ is similar except that we introduce $0$ in the sequence : $x_{1:N_{a,\epsilon}}$.
\end{proof}

\section{Proof of Lemma \ref{lemmax}}
Let $t_{n} = 5 \sigma_{\star}^{2} \log n $. One has
\begin{eqnarray*}
\bigP_{\theta^{\star}}\big(|Y|^{2}_{(n)} \ge   t_{n}  \big) &\leq & \max_{x_{1:n}\in\bigX^{n}}\bigP_{\theta^{\star}}\left(|Y|^{2}_{(n)} \ge t_{n}\vert X_{1:n}=x_{1:n} \right)\\
 & = & \max_{x_{1:n}\in\bigX^{n}} \left\{ 1- \prod_{i=1}^{n}\bigP_{\theta^{\star}}\left(Y_{i}^{2}\le t_{n} \vert X_{i}=x_{i}  \right)\right\}\\
 & \le & 1-\left[\bigP\left( U^{2} \le \frac{t_{n} - M }{\sigma_{\star}} \right)\right]^{n}\\
\end{eqnarray*}
where $M = \max_{i=1,...,k} m_{i}^{\star}$ and $U$ is a Gaussian random variable with distribution $\mathcal{N}(0,1)$.
Then, for large enough $n$ : 
\[\bigP_{\theta^{\star}}\left(|Y|^{2}_{(n)} \ge t_{n}  \right)  \le \frac{1}{n^{3/2}} \]
and the result follows  from Borel Cantelli Lemma. \cite{Teicher}

\bibliographystyle{IEEEtran}
\bibliography{refvlhmm}

\end{document}